\documentclass[final,11pt]{elsarticle}
\makeatletter
\renewcommand*\l@subsubsection[2]{{}{}{}} %
\makeatother

\usepackage[utf8]{inputenc}
\usepackage{amssymb}
\usepackage{amsmath}
\usepackage{amsfonts}
\usepackage{setspace}
\usepackage{stackrel}

\usepackage{subcaption}

\usepackage{amsthm}
 \usepackage[pageanchor=false,colorlinks=true,linkcolor=blue,citecolor=blue,urlcolor=blue]{hyperref}

\usepackage[normalem]{ulem} 

\newcommand\bR{\mathbb{R}}

\newcommand\bC{\mathbb{C}}

\newcommand{\tspan}{\mathrm{span}}

 \newcommand{\Fix}{\mathrm{Fix}}

 \newcommand{\supp}{\mathrm{supp}}

 \newcommand{\ls}{\langle}
 \newcommand{\rs}{\rangle}
   \newcommand{\re}{\mathcal{R}e}

\newtheorem{theorem}{Theorem}[section]

 \newtheorem{lemma}{Lemma}[section]
 \newtheorem{definition}{Definition}[section]
\newtheorem{remark}{Remark}[section]

\newcommand{\gt}{\hat{x}}

\journal{}
\begin{document}
\begin{frontmatter}

\title{Towards optimal algorithms for the recovery of low-dimensional models with linear rates}
\author{Yann Traonmilin\footnote[1]{Corresponding author. E-mail address: \url{yann.traonmilin@math.u-bordeaux.fr} } and Jean-François Aujol and Antoine Guennec}

\affiliation{organization={Univ. Bordeaux, Bordeaux INP, CNRS, IMB, UMR 5251},
	postcode={F-33400}, 
	state={Talence},
	country={France}}

\begin{abstract}We consider the problem of recovering elements of a low-dimensional model from linear measurements. From signal and image processing to inverse problems in data science, this question has been at the center of many applications. Lately, with the success of models and methods relying on deep neural networks, there has been a  multiplication of different algorithms and recovery results. Comparing  the performance of recovery algorithms becomes a complex task without a unifying framework. In this article, as a first step for the study of general algorithms for low-dimensional recovery,  we study  a  class of generalized projected gradient descent algorithms that can recover a given low-dimensional model with linear rates. The obtained rates decouple the impact of the quality of the measurements with respect to the model from the geometry of the properties of the chosen generalized projection: we can directly measure  performance  through a restricted Lipschitz constant of the projection with respect to the low dimensional model. By optimizing this constant, we define an optimal generalized projected gradient descent. Our general approach provides an optimality result in the case of sparse recovery. Moreover, our framework allows for a common interpretation of linear rates of recovery in the context of both sparse models and  models induced by some ``plug-and-play''  imaging methods that rely on deep neural networks. These rates of recovery  are observed in experiments on synthetic and real data.
\end{abstract}

\end{frontmatter}

\section{Introduction}

In this paper, we consider the general  noiseless observation model in finite dimension:
\begin{equation}\label{eq:model_obs}
 y = A \gt
\end{equation}
where $y \in \bC^m$ is a vector of measurements, $A$ is an under-determined linear map from $\bC^N$ to $\bC^m$ (i.e. $m <N$) and $\gt \in \bC^N$ is the unknown vector we want to recover.  The problem of recovering $\gt$ from $y$ is ill-posed. It is thus necessary to use prior information on $\gt$ to hope for an estimation of $\gt$ with a guarantee of success.

In this work, we will suppose that $\gt$ belongs to a homogeneous set $\Sigma$  (for every $x \in \Sigma$ and $\lambda \in \bC$, $\lambda x \in \Sigma$) that models known properties of the unknown.  In the sparse recovery literature, such sets $\Sigma$  are often considered to be (potentially infinite) union of (linear) subspaces.  Examples of such models include sparse as well as low-rank models (when $\bC^N$ is a space of matrices) and many of their generalizations (see~\cite{Foucart_2013} for an overview). The problem of recovering elements of a low-dimensional model from their measurements has been at the center of inverse problems in data science. This is for instance the case for many problems in signal and image processing where the unknown $\gt$ is the signal or image of interest, and the matrix $A$ models a degradation such as a subsampling or a blur. Note that, in this article, we consider the \emph{noiseless} case as we focus on the geometry of the estimation of elements of $\Sigma$ in relation to the chosen algorithm and the measurement operator $A$.

To recover $\gt$, a classical method is to solve the  minimization problem
\begin{equation} \label{eq:minimization1}
 x^\star \in \arg\min_{x \in \bC^N} \frac{1}{2}\|Ax-y\|_2^2  + \lambda R(x)
\end{equation}
where $R$ is a function (the regularizer) that  enforces some structure on the solution $x^\star$. It has been shown that functions $R$ of the form $d(\cdot, \Sigma)$ (distance to the model set for a given norm) achieve the best identifiability guarantees~\cite{bourrier2014fundamental}. Using the characteristic function of $\Sigma$ for $R(x)$ yields the constrained minimization 
\begin{equation} \label{eq:minimization_const}
	x^\star \in \arg \min_{x\in\Sigma} \frac{1}{2}\|Ax-y\|_2^2  
\end{equation}
 Unfortunately, for classical low-dimensional models such as sparse models, such optimization problems are NP hard~\cite{huo2010complexity}.

Another possibility is to use a convex proxy of minimization~\eqref{eq:minimization1} (i.e. using a convex $R$) that guarantees the recovery of elements of $\Sigma$. A general method to choose the best possible convex regularization $R$ (with respect to recovery guarantees of $\Sigma$ with $R$) has been presented in \cite{traonmilin2021theory}. However, in terms of practical recovery guarantees,  one must combine guarantees of success of~\eqref{eq:minimization1} (i.e. $x^\star = \hat{x}$) with convergence guarantees of the chosen algorithm.

While the convergence of many algorithms is verified for convex problems, heuristics approximating a non-convex minimization of~\eqref{eq:minimization1} (often called non-convex methods) are generally studied directly, and identifiability guarantees of the chosen algorithm are proven under conditions much more stringent than the guarantees of the ideal non-convex minimization. 

In this landscape of methods for solving inverse problems, algorithms based on deep priors have become the state-of-the-art reference in domains such as image processing. These methods have in common the fact that the prior used, one way or another, is learned on a large database with a deep neural network (DNN). Such priors are called deep priors. A large part of the resulting algorithms,  such as plug-and-play methods (PnP) \cite{venkatakrishnan2013plug}, fall within a non-convex framework, which leads to convergence to local minima with sublinear rates. In~\cite{liu2021recovery}, linear convergence requires global Lipschitz constants that are not verified for classical sparse recovery problems. Moreover, while the impact  of the properties of the low-dimensional model is studied in some works on deep priors, it is not extensively studied in the plug-and-play literature.

In this paper, we propose to study directly the identifiability performance of a class of recovery algorithms in order to answer the question:

\vspace*{1em}

\begin{center}
\emph{Given a low-dimensional model $\Sigma$, what is the optimal algorithm (in a given specific class of algorithms) to recover elements of $\Sigma$ from linear observations?}
\end{center}
\vspace*{1em}

We focus on this question for a class of algorithms that can be thought of as generalized projected gradient descent methods.  In this context,   the problem of low-dimensional recovery can be studied beyond variational methods where a functional (to be minimized)  must be explicit (in particular a regularizer). We propose state-of-the-art guarantees for these methods. Once guarantees are obtained, we can define the optimal method as the one having the best guarantees. Moreover, the framework allows us to consider algorithms relying on deep priors and to understand their identifiability and convergence properties.

\subsection{Generalized projected gradient descent}

We define the class of generalized projected gradient descent (GPGD) methods of interest as the iterations:

\begin{equation}\label{eq:def_GPGD}
	x_{n+1} = P_{\Sigma}(x_n) - \mu A^H(A P_\Sigma(x_n)-y)
\end{equation}
where  $P_\Sigma$ is a generalized projection (a (set-valued) map into $\Sigma$) and $\mu >0$ is a step size. In the context of sparse recovery with the orthogonal projection onto the set of $k$-sparse vectors, GPGD corresponds to the iterative hard thresholding (IHT) algorithm \cite{blumensath2010normalized} (or more generally iterative thresholding pursuit  \cite{foucart2011hard}).  In the context of plug-and-play methods, the generalized projection $P_\Sigma$ is a general purpose denoiser realized by a deep neural network~\cite{zhang2021plug}.

Note that projected gradient descent is often defined with the order of projection and gradient steps being reversed compared to GPGD defined in~\eqref{eq:def_GPGD}.  We will keep the name generalized projected gradient descent in the following to avoid confusion. In this article, we will specifically give theoretical results about GPGD. However, it will be useful to think of GPGD as a subset of a more general class of algorithms to lay out further questions about optimal low dimension recovery in a purely algorithmic framework:  methods of averaged directions (see Section~\ref{sec:av_direc}). In this setting, it is not necessary to define the $\ell^2$-datafit to design GPGD (this will be discussed in more details in Section~\ref{sec:choice_datafit}).

\subsection{A notion of optimal algorithm for low-dimensional recovery}

 In~\cite{traonmilin2021theory}, a notion of optimal (convex) regularization for low-dimensional recovery was proposed. Many possibilities, depending on the desired objective, exist for the definition of optimality. In this article, we will propose an optimality condition that provides a trade-off between the identifiability of elements of $\Sigma$ and the rate of convergence for newly given state-of-the-art recovery guarantees.

 In the case where elements $\gt \in \Sigma$ are uniquely identifiable from $y$ (i.e. there is a unique $\hat{x} \in \Sigma$ such that $A\hat{x}=y$), we  consider algorithms  where we can guarantee \emph{uniform recovery} with a global convergence, i.e. given $A$ (with the adequate properties) and any $\hat{x} \in \Sigma$~:

\begin{equation}
 \|x_n -\gt\|_2  \to_{n \to + \infty} 0
\end{equation}
for any choice of initialization $x_0$. Note that if any $\gt \in\Sigma$ is identifiable from $A\gt$, then the operator $A$ necessarily has a lower restricted isometry property (RIP)~\cite{bourrier2014fundamental}~(a property that will be central in our analysis, see Definition~\ref{def:RIC}).

Within the chosen set of algorithms, we will look for algorithms achieving linear rates of convergence, under a RIP assumption on $A$, i.e. there is $r \in [0;1) $ such that, for all $n \geq 0$,
\begin{equation}\label{eq:linear_rate}
 \|x_{n+1} - \gt\|_2 \leq r \|x_n - \gt\|_2.
\end{equation}
In other words, we specifically search for algorithms that perform global fast recovery of the elements of the model (as was demonstrated for iterative thresholding algorithms for sparse recovery with a restricted isometry  of $A$ \cite{blumensath2010normalized}).

Specifically, we ask ourselves what choice of generalized projection $P_\Sigma$ in GPGD  yields the best rates and recovery guarantees.   In this study,  we do not take into account the complexity of the computation of the projection $P_\Sigma$.  Although it is often fast to compute in the case of sparse modeling (in that case, projection is just a hard thresholding operation),  there are other instances when it can be much more demanding (i.e. projection onto sets of low-rank tensors). For plug-and-play methods with deep priors, the projection is realized by a forward pass in a deep neural network, which is a fast operation. Our contribution also opens broader questions about optimality within the class of averaged directions algorithms that will be discussed in Section~\ref{sec:av_direc}.

\subsection{Contributions}

In this article, we propose a framework for the study of the optimality of low-dimensional recovery with generalized projected gradient descent.

\begin{itemize}
 \item In Section~\ref{sec:recovery}, we propose a linear rate of recovery for generalized projected gradient descents. The obtained linear rate decouples the quality of the measurements through the restricted isometry constant of $A$ and the quality of the chosen algorithm through a newly introduced \emph{restricted} $\beta$-Lipschitz condition on $P_\Sigma$.  This enables the introduction of a well defined optimality notion for the class of projected gradient descent algorithms relying on the minimization of the \emph{restricted} Lipschitz constant $\beta$. 

 \item  In Section~\ref{sec:opt_proj}, we give properties related to  the \emph{restricted} Lipschitz condition. In particular, the orthogonal projection $P_\Sigma^\perp$ onto general sets (when it exists) always satisfies this condition for $\beta=2$. In the context of sparse recovery, we show that the orthogonal projection is indeed optimal when considering the family of model sets $\Sigma_k$ for all levels of sparsity $k$. This result also shows that for fixed sparsity, while the optimal projection may not be the orthogonal projection, its \emph{restricted} Lipschitz constant is close to the constant of the orthogonal projection. Also note that finding projections with optimal \emph{restricted} Lipschitz constants is of mathematical interest independently of the context of low-dimensional recovery.

 \item In Section~\ref{sec:av_direc}, we make the link between generalized projected gradient descent and iterative algorithms where a data-fit direction and a regularizing direction are averaged. We discuss open questions about  design and optimality within  these larger classes of algorithms in relation to our results.

 \item  In Section~\ref{sec:deep_priors}, we interpret our results in the context of plug-and-play methods for solving inverse problems with deep priors. We show  that, when we explicitly consider  the underlying low-dimensional model  and the corresponding projection that is performed, the so-called proximal gradient descent plug-and-play method (that can be interpreted as a generalized projected gradient descent) exhibits linear rates of convergence towards elements of the model induced by the chosen denoiser if the \emph{restricted}-Lipschitz property is verified.  We show experimentally this global linear convergence and verify  that the  \emph{restricted}-Lipschitz property is indeed satisfied empirically for some general purpose denoisers (DRUNet denoisers in our case).

\end{itemize}

\subsection{Related work}\label{sec:related_work}

Our work aims at building a framework for the design of optimal algorithms for the recovery of low-dimensional models. This follows some ideas from~\cite{traonmilin2021theory} where optimal convex regularizers are defined and calculated. In~\cite{leong2022optimal}, the authors propose a similar framework to design a possibly non-convex regularizer from a data source. Another direction of research related to optimal methods is studying the intrinsic computational complexity of recovery algorithms in the context of inverse problems \cite{cucker1999complexity,burgisser2013condition,roulet2017computational,bastounis2021extended}. Our approach mainly relies on restricting the class of considered algorithms to be able to find non-trivial optimal algorithms. Optimality in the context of function minimization has been a broad question. An approach is to find algorithms achieving "worst-case" rates theoretically (see e.g. \cite[Sec. 2.2]{nesterov2018lectures}) or even numerically~\cite{goujaud2024pepit}.

Projected gradient descent, has been widely studied in many applications. For sparse recovery, projected gradient descent is called iterative hard thresholding (or more generally hard thresholding pursuit) and has been shown to linearly converge to the unknown under a restricted isometry property of the observation operator $A$ \cite{blumensath2010normalized,foucart2011hard}. For the recovery of low-dimensional models, conditions for the linear convergence of projected gradient descent with approximate projections are given in \cite{golbabaee2018inexact}. In \cite{bahmani2016learning}, global convergence of PGD is given for a class of generalized sparsity models and the orthogonal projection. In \cite{tirer2021convergence}, general linear convergence of PGD with the orthogonal projection is shown for constrained minimization. Thanks to our convergence analysis, our work explicitly links the convergence rate of PGD with a restricted Lipschitz constant of the considered  (general) projection. Within thresholding algorithms, \cite{liu2020between} studies optimal thresholding operators with respect to a local concavity property used to give local convergence properties of projected gradient descent for minimizing general functions in \cite{barber2018gradient}. Another local convergence analysis is provided in \cite{vu2022asymptotic}.  In our work, we only consider global convergence. In \cite{kamilov2016learning}, it is proposed to learn an optimal non-linearity with a data-driven method. Note that global convergence of gradient projection has been shown under a general KL property in \cite{attouch2013convergence}. General stationary properties of the iterates of PGD are given in~\cite{olikier2024projected}. In this work, we focus on linear rates of convergence to solutions of linear inverse problems.

One objective of this paper is to give insights into the geometry of algorithms using deep priors to solve inverse problems, mainly in imaging. The recent literature on the subject is  profuse (see \cite{ongie2020deep,scarlett2022theoretical} for an overview). In particular, many variations of plug-and-play methods (where a deep neural network is used to perform an iteration of an optimization algorithm) exist, each one corresponding to a variation of an optimization algorithm such as forward-backward and ADMM \cite{venkatakrishnan2013plug,cohen2021regularization,chen2021deep,kamilov2023plug}. Specific designs of the regularizing direction have been given to guarantee the sublinear convergence of such methods under a global Lipschitz condition~\cite{hauptmann2024convergent} and a differentiability hypothesis on the functional to minimize~\cite{hurault2022gradient}. In \cite{liu2021recovery}, a decomposition of the rate of convergence between a restricted isometry condition and a global Lipschitz condition (which is not verified for sparse recovery) is shown. The control of the global Lipschitz constant of the DNN architecture is also at the center of other works~\cite{unser2024parseval}.  Our work shows that the control of a weaker restricted Lipschitz constant is sufficient for low-dimensional recovery. We must also cite methods where the projection onto the model set $\Sigma$ is explicitly using, e.g auto-encoders \cite{peng2020solving} or other generative models such as variational auto-encoders~\cite{gonzalez2022solving}. Recovery guarantees of low-dimensional models are studied under very stringent conditions (Gaussian measurements) in~\cite{hand2019global}. In \cite{tachella2023sensing}, a general study of the identifiability of models obtained in a learning context is performed.

\subsection{Notation, definitions and preliminaries} \label{sec:notations}

Given a linear operator $A : \bC^N \to \bC^m$, we denote by $A^H$ its Hermitian adjoint. The identity operator is denoted by $I$. Given a function $f : \bC^N \to \bC^N$, we call $\Fix (f) = \{x \in \bC^N : f(x) =x \}$ the set of fixed points of $f$.

We use the following definition of a restricted isometry constant.

\begin{definition}\label{def:RIC}
The operator $B$  has restricted isometry constant $\delta \in [0,1)$ on the secant set $\Sigma-\Sigma =\{x_1-x_2 : x_1,x_2 \in \Sigma \}$  if for all $x_1,x_2 \in \Sigma$

\begin{equation}
 \|(I-B)(x_1-x_2)\|_2\leq \delta \|x_1-x_2\|_2
\end{equation}
We denote by $\delta_\Sigma(B)$ the  smallest restricted isometry constant (RIC) of $B$.
\end{definition}

Note that $\delta_\Sigma(B)$ is bounded by the operator norm of $B-I$.  In this article, operators $B$ will be of the form $B= \mu A^HA$ where $\mu$ is the gradient step size.
For $B= A^HA,$ the inequality $\delta_\Sigma(A^HA)<1$  implies a restricted isometry property (RIP) of the operator $A$ in the traditional sense defined by Equation~\eqref{eq:def_RIP}~\cite{Foucart_2013} as shown in the following Lemma.

\begin{lemma} Let $\Sigma \subset \bC^N$. If $A^HA$ has restricted isometry constant $ \delta_\Sigma(A^HA)$ and $x_1,x_2 \in \Sigma$, we have
 \begin{equation}\label{eq:def_RIP}
  (1-\delta_\Sigma(A^HA))\|x_1-x_2\|_2^2\leq\|A(x_1-x_2)\|_2^2 \leq (1+\delta_\Sigma(A^HA))\|x_1-x_2\|_2^2.
\end{equation}
\end{lemma}

\begin{proof}
We have, for $x_1,x_2 \in \Sigma$, $x_1 \neq x_2$,
\begin{equation}
\begin{split}
\left| \frac{\|A(x_1-x_2)\|_2^2 -\|x_1-x_2\|_2^2}{\|x_1-x_2\|_2^2} \right|&=  \left|  \frac{\ls A(x_1-x_2), A (x_1-x_2) \rs  -\ls x_1-x_2, x_1-x_2 \rs   }{\|x_1-x_2\|_2^2}\right|\\
&=  \left|  \frac{\ls(A^HA-I)(x_1-x_2), x_1-x_2 \rs  }{\|x_1-x_2\|_2^2}\right|
 \end{split}
\end{equation}
Using the Cauchy-Schwarz inequality, and the definition  of $\delta_\Sigma(A)$, we have 
\begin{equation}
\begin{split}
 \left| \frac{\|A(x_1-x_2)\|_2^2 -\|x_1-x_2\|_2^2}{\|x_1-x_2\|_2^2} \right|&\leq \frac{\|(A^HA-I)(x_1-x_2)\|_2  }{\|x_1-x_2\|_2}\leq \delta_\Sigma(A).\\
 \end{split}
\end{equation}

\end{proof}

It is sometimes useful to  consider operators $A$ that have been scaled to have the best possible RIC $\delta$. This eliminates the typical problem with the RIP hypothesis that multiplying the measurement operator by a factor does not change recovery capabilities but can worsen the RIC.

\begin{definition}
 Suppose that $B$ has a RIC $\delta_\Sigma(B)$. We say that $B$ is optimally scaled for the RIC  if for all $\lambda \in \bR$ such that $\delta_\Sigma(\lambda B)$ exits, we have $\delta_\Sigma(B) \leq \delta_\Sigma(\lambda B)$.
\end{definition}

Hence, if $B$ is not optimally scaled, we can consider instead the scaled operator $\tilde{B} = \lambda_0 B$ such that $\delta_\Sigma(\tilde{B}) = \inf_{\lambda \geq 0} \delta_\Sigma(\lambda B)$.

Note that the lower RIP (left inequality in \eqref{eq:def_RIP}) is a necessary condition for the identifiability of all elements of $\Sigma$ from measurements with the operator $A$ \cite{bourrier2014fundamental}.  Consequently, in the context of uniform recovery of $\Sigma$ from linear measurements in finite dimension, we can always make  a RIC hypothesis on an optimally scaled operator $A^HA$  (thanks to  the lower RIP and boundedness of the operator norm,  see e.g. \cite[Lemma 3.1]{traonmilin2021theory}), otherwise such uniform recovery is not possible.

We now define generalized projections and orthogonal projections.

\begin{definition}[Generalized projection]\label{def:proj}
Let $\Sigma \subset \bC^N$. A (set-valued) generalized projection onto $\Sigma$ is  a (set-valued) function $P$ such that for any $z\in\bC^N$, $P(z) \subset \Sigma$.
\end{definition}

By abuse of notation, to facilitate reading, an equation true for any $w \in P(z)$ is written using the notation $P(z)$. We introduce orthogonal projections (metric projections for the $\ell^2$ norm) on sets where they exist.
 
\begin{definition}[Proximinal sets and orthogonal projections]\label{def:orth_proj}
	Let $\Sigma \subset \bC^N$. The set $\Sigma$ is proximinal if for all $z \in \bC^N$, we have 
	\begin{equation}
	  \left( \arg \min_{x\in\Sigma} \| x-z \|_2   \right) \neq \emptyset
	\end{equation}
	
	Now suppose $\Sigma$ is a proximinal set, we define the orthogonal projection onto $\Sigma$ as 
	\begin{equation}
		P_\Sigma^\perp(z)  := \arg \min_{x\in\Sigma} \| x-z \|_2.
	\end{equation}
	Notice that $P_\Sigma^\perp(z)$ may be set-valued.
\end{definition} 

The proximinal property (used e.g. in \cite{franchetti1986embedding})  is verified for finite union of linear subspaces such as sparse models. It is also verified for spaces of low rank matrices with the singular value thresholding operator.

\section{On optimal low-dimensional recovery generalized projected gradient descent}\label{sec:recovery}

In this section, we  give our main theorem guaranteeing recovery with generalized projected gradient descent. This general recovery result is the basis for our definition of optimal recovery. We show that the  class of  generalized projected gradient descent can recover low-dimensional models provided that the projection satisfies a \emph{restricted} Lipschitz condition. This condition for recovery with a linear rate is weaker than the classical Lipschitz condition.

\begin{definition}[Restricted Lipschitz property]\label{def:lip_const}
Consider a generalized projection $P$. Then $P$ has the restricted $\beta$-Lipschitz property with respect to $\Sigma$ if for all $z \in \bC^N, x \in \Sigma,$, we have
\begin{equation}
\begin{split}
 \|P(z)-x\|_2 &\leq \beta \|z-x\|_2\\
 \end{split}
\end{equation}
We note $\beta_{\Sigma}(P)$ the smallest $\beta$ such that $P$ has the restricted $\beta$-Lipschitz property.

For set-valued projections we say that $P$ has the restricted $\beta$-Lipschitz property if for all  $z \in \bC^N, x \in \Sigma, u \in P(z)$, $\|u-x\|_2 \leq \beta \|z-x\|_2$.

\end{definition}

Note that any $P$ verifying this condition is such that $\Sigma \subset \Fix (P)$ as  $\|z-x\|_2 =0$  (i.e. $z=x$) implies $P(z) =x$ (see also Lemma~\ref{lem:charac_lip1}). Hence, this is a weaker statement than the classical Lipschitz property $\|P(z)-P(x)\|_2 \leq \beta \|z-x\|_2$, since one variable is restricted to $\Sigma$. 

We expect that a restricted Lipschitz constant is such that $\beta \geq 1 $ (with equality for the orthogonal projection onto a linear subspace, see Lemma~\ref{lem:lip_const_subspace}). Note that, in a context where orthogonal projections onto \emph{convex} sets are considered for the regularization by denoising algorithm ($g_n = (x_n -P(x_n))$), this restricted Lipschitz constant  was identified as important in the analysis of the plug-and-play regularization-by-denoising (RED) algorithm but not explicitly defined as a minimal hypothesis for convergence \cite{cohen2021regularization}. The main advantage of this restricted Lipschitz property is to build a common framework between the sparse recovery setting where the projection is generally not globally Lipschitz (e.g. hard thresholding) and denoisers used as generalized projections.

In  Section~\ref{sec:opt_proj}, the existence of  projections with the restricted $\beta$-Lipschitz property is guaranteed on general sets: the orthogonal projection onto a proximinal set is restricted $2$-Lipschitz.

We can now give a general result that guarantees recovery with a linear rate.

\begin{theorem}[Linear recovery of low-dimensional models]\label{th:gen_convevergence}
Let $\Sigma \subset \bC^N$ and $P_\Sigma$ a generalized projection onto $\Sigma$. Suppose the operator $\mu A^H A$ has RIC $ \delta:=\delta_\Sigma(\mu A^HA) <1$.
For any $\gt \in \Sigma, x_0 \in\bC^N$, consider the iterates $x_n$ resulting from Algorithm~\eqref{eq:def_GPGD}. Suppose  $P_\Sigma$ has the restricted $\beta$-Lipschitz property with respect to $\Sigma$ (with $\beta := \beta_\Sigma(P_\Sigma)$). 
Then we have, for all $\gt\in\Sigma, x_0 \in \bC^N, n\geq 1$,

\begin{equation}\label{eq:res_linear_rate}
 \|x_n-\gt\|_2 \leq (\delta\beta)^n \|x_0 -\gt\|_2.
\end{equation}

Moreover, if $\delta \beta <1$, then \eqref{eq:res_linear_rate} implies a linear convergence of $x_n$ to $\hat x$.

\end{theorem}
\begin{proof}
Let $n\geq 0$. With $y = A \gt$ and the definition of the RIC  $ \delta=\delta_\Sigma(\mu A^HA) <1$, using the fact that  $P_\Sigma(x_n)-\gt \in \Sigma-\Sigma$ by definition of $P_\Sigma$ and $\gt$, we have 
\begin{equation}
\begin{split}
	\|x_{n+1}-\gt\|_2 &= \|P_\Sigma(x_n) - \mu A^H(A P_\Sigma(x_n)-y) -\gt\|_2 \\
	&=  \| (I- \mu A^HA) (P_\Sigma(x_n)-\gt)\|_2 \\
	&\leq \delta\|P_\Sigma(x_n)-\gt\|_2.
\end{split}
\end{equation}
As $\gt \in \Sigma$, we use the restricted Lipschitz property of $P_\Sigma$  and deduce
\begin{equation}
	\begin{split}
		\|x_{n+1}-\gt\|_2 &\leq \delta\beta\|x_n-\gt\|_2.
		\end{split}
\end{equation}
The conclusion \eqref{eq:res_linear_rate} is obtained by a direct induction.
\end{proof}

This theorem tells us that GPGD guarantees the recovery of low-dimensional models with linear rates under a RIP condition provided a restricted Lipschitz property of the projection. Note that the optimal step size $\mu$ for these guarantees is the optimal rescaling of $A^HA$ for the RIC. 

In this context we can define optimality of GPGD for those guarantees. Optimal GPGD are those whose projection minimizes the restricted Lipschitz constant; it quantifies two properties of the algorithm:
\begin{itemize}
 \item the identifiability  properties of GPGD: if $\delta_\Sigma(\mu A^HA)<\frac{1}{\beta}$, then $x_n \to \gt$;
 \item the rate of convergence: for a fixed operator $A$ and step size $\mu$, the smaller the restricted Lipschitz constant $\beta_\Sigma(P_\Sigma)$, the faster the recovery of $\gt$. 
\end{itemize}

To illustrate the convergence guarantees of Theorem~\ref{th:gen_convevergence}, we take the example of recovery of $k$-sparse signals (that is developed in Section~\ref{sec:opt_iht}). We use  random Gaussian measurements.  In this case, it is well known that $\delta$ is decreasing with respect to the number of  measurements $m$. The use of GPGD with the orthogonal projection $P_{\Sigma_k}^\perp$ onto the set $\Sigma=\Sigma_k$ of $k$-sparse signals corresponds to the iterative hard thresholding algorithm. It is possible to worsen the restricted Lipschitz constant of  $P_{\Sigma_k}^\perp$ by using a perturbed version $P_\alpha$  such that $\beta_\Sigma(P_\alpha) \leq \beta_\Sigma(P_\Sigma^\perp) +\alpha$ with $\alpha >0$ (see \cite[Appendix B]{joundi2025sparse}). We show in Figure~\ref{fig:app_iht} the impact of $\beta_\Sigma(P_\alpha)$ on the identifiability of GPGD and the convergence rate in the noiseless setting ($m= 150, n=300$). These experiments match our two observations on recovery guarantees given by Theorem~\ref{th:gen_convevergence}.
	
	\begin{figure}
		\centering
		\includegraphics[width=.75\linewidth]{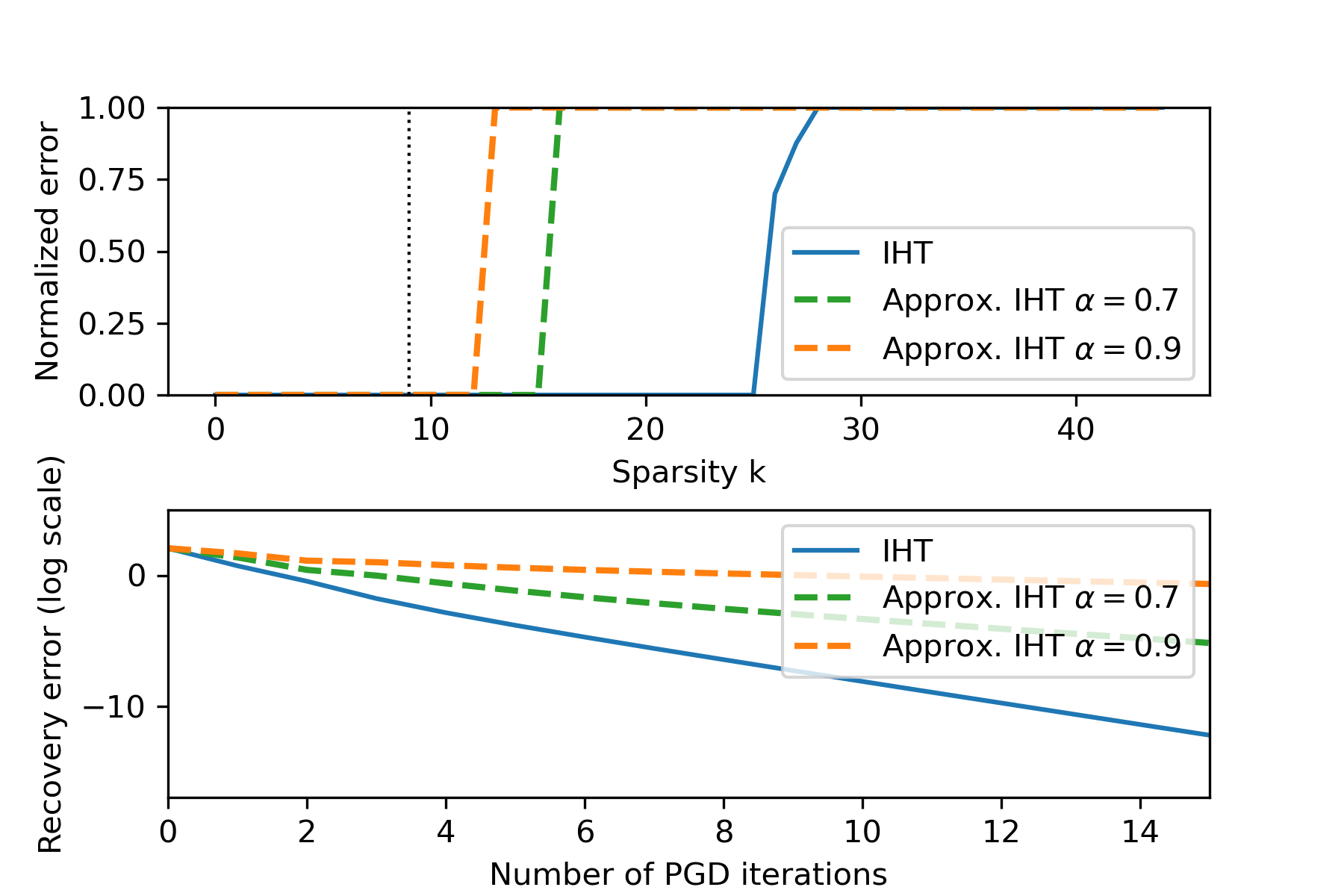}
		\caption{\textcolor{black}{Impact of $\beta_{\Sigma_k}(P_\alpha)$ in GPGD for sparse recovery. Top: Normalized error bound for 95$\%$ of the 50 experiments with respect to the sparsity $k$ of $\hat{x}$. Bottom: convergence of GPGP with $k=9$. We observe that increasing $\beta$ (through the parameter $\alpha$) decreases both the identifiability properties of GPGD and its convergence rate.}}
		\label{fig:app_iht}
	\end{figure}

Given these facts,  we propose a quantitative optimality measure (in terms of convergence) of a generalized projected descent algorithm  parametrized by a projection $P$ with $\beta_\Sigma(P)$.

\begin{definition}[Optimal projection]\label{def:optimal_proj}
We define the optimal projection $P^\star$ for the uniform recovery of a low-dimensional model set $\Sigma$ with generalized projected gradient descent with a uniform linear rate (given by Theorem~\ref{th:gen_convevergence}) as
\begin{equation}
P^\star \in \arg \min_{P \in \Pi_\Sigma} \beta_\Sigma(P)
\end{equation}
where $\Pi_\Sigma$ is the set of (generalized) projections onto $\Sigma$ having a restricted $\beta$-Lipschitz  property.
\end{definition}
In the next Section, we  investigate the  problem of finding optimal projections and focus on the \emph{orthogonal projection}. We notice that it is restricted $2$-Lipschitz for general proximinal sets.  In particular, we show that the orthogonal projection is near-optimal for sparse recovery. For closed homogeneous proximinal  sets, we show that an optimal projection exists (Theorem~\ref{th:existence}). We also will see that our convergence result can  be used to interpret the convergence of a certain class of plug-and-play algorithms for imaging inverse problems with deep priors (Section~\ref{sec:deep_priors}).

\begin{remark} The constant $\beta$ quantifies \emph{global} convergence. Indeed, for $\Sigma$ a finite union of linear subspaces, the orthogonal projection onto $\Sigma$ is locally restricted $1$-Lipschitz while globally restricted $2$-Lipschitz (see Section~\ref{sec:opt_proj}). Hence, the convergence rate parameter is just $\delta(\mu A^T A)$ if GPGD with the orthogonal projection is initialized close enough to the solution. For any finite union of subspaces and GPGD with the orthogonal projection, if we are able to initialize close enough to the solution (identify the linear subspace containing $\hat{x}$), the convergence rate does not depend on the geometry of the model and the restricted Lipschitz constant is locally optimal equal to $1$. Consequently, our notion of optimality is related to  \emph{global} convergence guarantees. 
\end{remark}

\begin{remark}\label{rem:tightness}
A question that naturally arises is the tightness of Theorem~\ref{th:gen_convevergence}. We observe that the rate of convergence should be tight if the worst initialization for the restricted Lipschitz condition of the projection $P_\Sigma$ matches the worst case for the RIC constant of $A$. While in some cases, it should be possible to construct such $A$, the relation between the RIC of $A$ and the properties of $P_\Sigma$ might be more intricate in general.
\end{remark}

\section{On optimal projections for generalized projected  gradient descent}\label{sec:opt_proj}

In this section, we  notice that the orthogonal projection on proximinal sets has  restricted Lipschitz constant $\beta=2$. The context of homogeneous spaces allows us to treat many generalized sparsity models such as group sparsity (without overlap) \cite{beck2019optimization} and sparsity in levels \cite{li2019compressed} and even low-rank recovery~\cite{davenport2016overview}. For sparse recovery, we give an optimality result for the orthogonal projection (which corresponds to iterative hard thresholding).

\subsection{General properties of the restricted Lipschitz constant}

We begin by giving the following technical Lemmas that lead to a worst case for the restricted Lipschitz condition (Lemma~\ref{lem:max_Qc}). The first one reformulates the expression of the restricted Lipschitz constant.

\begin{lemma}[Characterization of the restricted Lipschitz property]\label{lem:charac_lip1}
	Let $u, x \in \Sigma, z\in \bC^N, \beta >0$ and define for any $c>0$:
	\begin{equation}\label{eq:def_Q}
		Q_{c}(u,z,x) := \|u\|_2^2 -c \|z\|_2^2    -2 \re\ls u- cz,  x \rs  +  (1-c) \|x\|_2^2.
	\end{equation}
	Then a generalized  projection $P$ has  restricted Lipschitz constant $\beta > 0 $ (Definition~\ref{def:lip_const}) if and only if
	\begin{equation}\label{eq:conditition_Qb}
		\sup_{z\in\bC^N; \; u \in P(z); \; x\in\Sigma} Q_{\beta^2}(u,z,x) \leq 0.
	\end{equation}
	
\end{lemma}

\begin{proof}
	Let $z\in\bC^N,x\in \Sigma$ and $u \in P(z)$. We have:
	\begin{equation}
		\begin{split}
			&\|u-x\|_2^2 - \beta^2\|z -x\|_2^2 \\
			& = \|u\|_2^2 -\beta^2  \|z\|_2^2    -2 \re\ls u,  x \rs +2 \beta^2 \re\ls z,  x \rs  +  (1-\beta^2) \|x\|_2^2 \\		
			&= \|u\|_2^2 -\beta^2  \|z\|_2^2    -2 \re\ls u- \beta^2 z,  x \rs  +  (1-\beta^2) \|x\|_2^2 \\
			&=  Q_{\beta^2}(u,z,x) .
		\end{split}
	\end{equation}
	We deduce that $	\|u-x\|_2^2 \leq \beta^2\|z -x\|_2^2 $ if and only if $  Q_{\beta^2}(u,z,x) \leq 0$ and $P$ has restricted Lipschitz constant $\beta$ if and only if \eqref{eq:conditition_Qb} is verified. 
\end{proof}

In the general setting of proximinal sets $\Sigma$ (Definition~\ref{def:orth_proj}), we can maximize the function $Q_{\beta^2}(u,z,x)$ with respect to $x\in \Sigma$. If $\Sigma$ is a finite union of subspaces it is immediate that $P_\Sigma^\perp$ exists as it is the minimum of the finite number of projections onto the individual subspaces. For sparse recovery, i.e. $\Sigma_k = \{x \in \bC^N : \|x\|_0\leq k \}$, $P_\Sigma^\perp$ is the hard-thresholding operator. For some other low-dimensional models such as low-rank models, the union is infinite but the orthogonal projection still exists (singular value thresholding). We recall some properties of the orthogonal projection onto union of subspaces.
\begin{lemma}\label{lem:prop_orth_proj}
	Suppose $\Sigma$ is a proximinal homogeneous set. Then  for all $z \in \bC^N$,
	\begin{enumerate}
		\item there exists a linear subspace $W \subset \Sigma$ such that   $P_W^\perp(z) \in P_\Sigma^\perp(z)$ ;
		\item for all linear subspaces $V \subset \Sigma$, $\|P_W^\perp(z) - z\|_2^2 \leq \|P_V^\perp(z)-z\|_2^2$;
		\item for all linear subspaces $V \subset \Sigma$, $\|P_W^\perp(z)\|_2^2 \geq \|P_V^\perp(z)\|_2^2$;
		\item  $ \re\ls P_W^\perp(z),P_W^\perp(z) -z\rs =0  $.
		\item  $	P_\Sigma^\perp$ is homogeneous: for $\lambda \in  \bR$, $	P_\Sigma^\perp(\lambda z) = \lambda  P_\Sigma^\perp(z).$
		
	\end{enumerate}
	
\end{lemma}

\begin{proof}
	Let $u \in P_\Sigma^\perp(z)$. Take any  linear subspace $W \subset \Sigma$ such that $\tspan(u) \subset W$. By definition of the orthogonal projection, $\|u-z\|_2 \leq \|w-z\|_2$ for any $w \in W \subset \Sigma$, hence $P_W^\perp(z) =u$. The other properties are direct properties of the orthogonal projection onto a linear subspace.
\end{proof}

We give the following Lemma that explicitly calculates the maximization over $x\in\Sigma$ in condition~\eqref{eq:conditition_Qb}.

\begin{lemma}\label{lem:max_Qc}
	Let $\Sigma$ be a homogeneous proximinal set. Let  $z\in \bC^N, u \in \Sigma$. We have the following properties.
	\begin{itemize}
		\item If $c > 1$, let
		\begin{equation}
			x^\star \in P_\Sigma^{\perp} \left(\frac{u- c z}{1-c}\right).
		\end{equation}
		Then
		\begin{equation}
			Q_c(u,z,x^\star) = R_c(u,z) :=  \max_{x \in\Sigma} Q_c(u,z,x),
		\end{equation}
		where $Q_c$ is defined in Lemma~\ref{lem:charac_lip1}.
		We have  the following expressions of the maximum:
		\begin{equation} \label{eq:expr_Qc_opt}
			\begin{split}
				Q_c(u,z,x^\star)
				&=\|u\|_2^2 -c  \|z\|_2^2    +(c-1) \|  P_\Sigma^{\perp} \left(\frac{ u-cz }{1-c} \right)\|_2^2. \\
			\end{split}
		\end{equation}
		
		\item If $c= 1$, $Q_c(u,z,x)$ is upper bounded with respect to $x$ if and only if for all $x \in \Sigma$, $ \re\ls u- z,  x \rs  = 0$. In this case, for all $x \in \Sigma$,
		\begin{equation}
			\begin{split}
				Q_c(u,z,x) &=\|u\|_2^2 -  \|z\|_2^2.
			\end{split}
		\end{equation}
	\end{itemize}
\end{lemma}

\begin{proof}
	If $c=1$,  $Q_c(u,z,x) = \|u\|_2^2 -\|z\|_2^2 - 2\re \ls u-z,x \rs$ is an affine function of $x\in \Sigma$. Hence, it is upper bounded if and only if  for all $x \in \Sigma$, $ \re \ls u- z,  x \rs  = 0$.
	
	Now let $c>1$. First notice that $Q_c(u,z,x)$ is a quadratic form with respect to $x$ with a negative leading coefficient. Thus  $\sup_{x \in\Sigma} Q_c(u,z,x)< +\infty$. By definition (Equation~\eqref{eq:def_Q}), we have
	\begin{equation}
		\begin{split}
			Q_c(u,z,x) &= -2 \re\ls u- c z,  x \rs  +  (1-c) \|x\|_2^2  + C \\
			&= (1-c) \left\|x - \frac{ u- c z}{1-c} \right\|_2^2 +C'
		\end{split}
	\end{equation}
	where $C,C'$ are constants that do not depend on $x$. Removing the constant terms, we have that maximizing $Q_c$ is equivalent to maximizing $\tilde{Q}_c$ on $\Sigma$, with
	\begin{equation}
		\tilde{Q}_c (x) := (1-c) \left\|x - \frac{ u- c z}{1-c}\right\|_2^2.
	\end{equation}
	As $1-c<0$, this is exactly the minimization of $\|x - \frac{ u- c z}{1-c}\|_2^2$ with respect to $x \in \Sigma$ which yields (as the orthogonal projection onto $\Sigma$ was supposed to exist) the optimal
	
	\begin{equation}
		x^\star =P_\Sigma^{\perp} \left(\frac{ u- c z}{1-c}\right).
	\end{equation}
	
	We deduce, using the definition of $Q_c$ (Equation~\eqref{eq:def_Q}), that for any $z \in \bC^N$,
	\begin{equation}
		\begin{split}
			&Q_c(u,z,x^\star)\\ 
			&=\|u\|_2^2 -c  \|z\|_2^2    -2 \re\ls u- c z,  P_\Sigma^{\perp} \left(\frac{ u- c z}{1-c}\right)\rs  +  (1-c) \|P_\Sigma^{\perp} \left(\frac{ u- c z}{1-c}\right)\|_2^2 \\
			&=\|u\|_2^2 -c  \|z\|_2^2    -2(1-c) \re\ls \frac{u- c z}{1-c},  P_\Sigma^{\perp} \left(\frac{ u- c z}{1-c}\right)\rs  +  (1-c) \|P_\Sigma^{\perp} \left(\frac{ u- c z}{1-c}\right)\|_2^2.\\
		\end{split}
	\end{equation}
	With Lemma~\ref{lem:prop_orth_proj}, for any $w \in \bC^N$, $\re\ls P_\Sigma^\perp(w),w \rs = 
	\|P_\Sigma^\perp(w)\|_2^2 + \re\ls P_\Sigma^\perp(w),w - P_\Sigma^\perp(w)\rs  = \|P_\Sigma^\perp(w)\|_2^2$ (the projection direction is orthogonal to the projection). We deduce
	
	\begin{equation}
		\begin{split}
			Q_c(u,z,x^\star) =R_c(u,z)
			&=\|u\|_2^2 -c  \|z\|_2^2    -2(1-c) \|  P_\Sigma^{\perp} \left(\frac{ u- c z}{1-c}\right)\|_2^2 \\
			& +  (1-c) \|P_\Sigma^{\perp} \left(\frac{ u- c z}{1-c}\right)\|_2^2 \\
			&=\|u\|_2^2 -c  \|z\|_2^2    +(c-1) \|  P_\Sigma^{\perp} \left(\frac{ u- c z}{1-c}\right)\|_2^2. \\
		\end{split}
	\end{equation}
	This is exactly conclusion~\eqref{eq:expr_Qc_opt}. 	
\end{proof}

In the hypotheses of Lemma~\ref{lem:max_Qc}, the study of the function $R_{\beta^2}(P(z),z)$ suffices to study the restricted Lipschitz constant of a projection $P$.
The next Lemma  shows in  the case of linear subspaces that the best restricted Lipschitz constant $\beta =1$ is reached by the orthogonal projection.

\begin{lemma}\label{lem:lip_const_subspace}
	Let $\Sigma = V \subset \bC^N$ be a linear subspace. Then $P_V^\perp$ has the restricted $1$-Lipschitz property.
\end{lemma}
\begin{proof}
	Let $z \in \bC^N, x \in V$. We have $\re\ls P_V^\perp(z)-z, x \rs =  0 $ (because $P_V^\perp(z)-z \perp V$).  Hence, with Lemma~\ref{lem:max_Qc}, we have  $Q_1(P_V^\perp(z),z,x^\star)=\|P_V^\perp(z)\|_2^2 -\|z\|_2^2 = -\|P_V^\perp(z)-z\|_2^2  \leq 0$. With Lemma~\ref{lem:charac_lip1}, $P_V^\perp$ is restricted $1$-Lipschitz.
	
	An alternate proof is to remark that this Lemma is a direct consequence of the fact that the orthogonal projection onto a convex closed subset is globally $1$-Lipschitz.
\end{proof}

The best possible restricted Lipschitz constant $\beta=1$ is never reached when  $\Sigma \varsubsetneq \bC^N$  and $\tspan(\Sigma)=\bC^N$ (e.g in the sparse or low-rank case). Hence, a constant $c>1$ is expected in challenging cases.
\begin{lemma}	
	Let $\Sigma \subset \bC^N$ be a homogeneous set such that  $\Sigma \neq \bC^N$ and $\tspan(\Sigma)=\bC^N$. For any projection $P$ onto $\Sigma$, $\beta_{\Sigma}(P)  >1$.\\
\end{lemma}
\begin{proof}
	By contradiction, assume that $P:\bC^N\rightarrow \Sigma$  has the $1$-Lipschitz property. Let $z \in\bC^N \setminus \Sigma$. From Lemma \ref{lem:max_Qc}, $Q_1(P(z), z, x)$ is upper bounded with respect to $x$ only if for all $x\in \Sigma$, $\re\langle P(z) - z, x\rangle =0 $. However, since $\text{Span}(\Sigma)=\bC^N$, this extends to $ \forall w\in \mathbb C^N$: just write $w= \sum_i \lambda_i x_i$ with $x_i \in\Sigma$. Then
	\begin{equation}
		\re\langle P(z)-z, w\rangle=\re\langle P(z)-z, \sum_i \lambda_i x_i\rangle=\sum_i \lambda_i\re\langle P(z)-z,  x_i\rangle =0.
	\end{equation}
	Take $w =P(z)-z$, we deduce that $\|P(z)-z\|_2^2 = 0$ and  $ z=P(z) \in \Sigma$. Since we supposed $z\notin \Sigma$, we reach a contradiction.
\end{proof}

In the following, we notice that the orthogonal projection always has a restricted Lipshitz property for proximinal sets.

\begin{lemma}\label{th:lip_const_orth_proj} Let $\Sigma$ be a proximinal set. Then  $P_\Sigma^\perp$ has restricted Lipschitz constant  $\beta = 2$.
\end{lemma}

\begin{proof}
	Let $z\in \bC^N, u \in P_\Sigma(z)$ and $x \in \Sigma $. With the triangle inequality, we have 
	\begin{equation}
		\|u-x\|_2 = \|u-z +z-x\|_2\leq \|u-z\|_2 + \|z-x\|_2.
	\end{equation}
	By definition of the orthogonal projection, as $u \in P_\Sigma(z)$ and $x\in\Sigma$, we have  $\|u-z\|_2\leq \|x-z\|_2$ and 
	\begin{equation}
		\|u-x\|_2 \leq \|x-z\|_2 + \|z-x\|_2 = 2 \|z-x\|_2.
	\end{equation}
\end{proof}

With Theorem~\ref{th:gen_convevergence}, this result shows the global linear convergence of projected gradient descent with orthogonal projection for a large class of low-dimensional models. Note as well that compressive measurements with RIC close to $0$ can be built for general low-dimensional models~\cite{Puy_2015}.  For more particular model sets such as sparse models, we can give a better estimation of the restricted Lipschitz constant which leads to an optimality result (when considering a collection of models) of the orthogonal projection. 

\begin{remark}  The upper bound on the restricted $\beta$-Lipschitz constant from Lemma \ref{th:lip_const_orth_proj} cannot be improved  for homogeneous sets. Consider  $\Sigma = V_{\varepsilon} \cup W_{\varepsilon} \subset \bR^2$, with  $\varepsilon >0$, $V_\epsilon = \{\alpha\cdot(1,\varepsilon)\in\bR^2\ |\ \alpha\in \bR \}$ and $W_\epsilon = \{\alpha\cdot(1,-\varepsilon)\in\bR^2\ |\ \alpha\in \bR \}$. Let $z=(1,0)$. Then, we have
	\begin{equation*}
		\begin{split}
			P_{V_{\varepsilon}}^\perp(z) &= \underset{w\in V_\varepsilon}{\arg\min}\|z - w\|^2 = \underset{\substack{w = \alpha\cdot (1,\varepsilon)^T \\ \alpha \in \bR}}{\arg\min}\ \|(1 , 0) -\alpha\cdot(1,\varepsilon)\|^2 \\
			&= \underset{\substack{w = \alpha\cdot (1,\varepsilon)^T \\ \alpha \in \bR}}{\arg\min}\  (\alpha - 1)^2 + \alpha^2 \varepsilon^2 = \frac{1}{1+\varepsilon^2}(1 , \varepsilon)
		\end{split}
	\end{equation*}
	and through the same computations, we have 
	\begin{equation*}
		P_{W_{\varepsilon}}^\perp(z) = \frac{1}{1+\varepsilon^2}(1 , -\varepsilon).
	\end{equation*}
	Then,
	\begin{equation*}
		\begin{split}
			\|P_{V_{\varepsilon}}^\perp(z) - z\|_2 &= \|P_{W_{\varepsilon}}^\perp(z) - z\|_2 = \sqrt{\left(1-\frac{1}{1+\varepsilon^2}\right)^2 + \left(\frac{\pm\varepsilon}{1+\varepsilon^2}\right)^2} \\
			&= \sqrt{\frac{\varepsilon^4 + \varepsilon^2}{(1+\varepsilon^2)^2}} = \frac{\varepsilon}{\sqrt{1+\varepsilon^2}} \\
			\end{split}
	\end{equation*}
	and thus $P_\Sigma^{\perp}(z) = \{P_{V_{\varepsilon}}^\perp(z), P_{W_{\varepsilon}}^\perp(z)\}$. Moreover, we have
	\begin{equation*}
		\|P_{V_{\varepsilon}}^\perp(z) - P_{W_{\varepsilon}}^\perp(z)\|_2= \frac{1}{1+\varepsilon^2} \|(1 , \varepsilon) - (1 ,  -\varepsilon)\|_2= \frac{2\varepsilon}{1+\varepsilon^2}.
	\end{equation*}
	Consequently,
	\begin{equation*}
	 \beta_\Sigma(P_\Sigma^\perp) \geq 	\frac{\|P_{V_{\varepsilon}}^\perp(z) - P_{W_{\varepsilon}}^\perp(z)\|_2}{\|z- P_{W_{\varepsilon}}^\perp(z)\|_2} = \frac{2}{\sqrt{1+\varepsilon^2}} \underset{\varepsilon\rightarrow 0}{\longrightarrow} 2. 
	\end{equation*}
\end{remark}

For closed homogeneous proximinal sets, we can prove the existence of an optimal projection. We first give the following result that allows us to use a continuity argument.
\begin{lemma}\label{lem:exist0b}
	Suppose that $\Sigma$ is a proximinal homogeneous set. Then the function $z \mapsto \|P_\Sigma^\perp(z)\|_2^2$ is continuous.
\end{lemma}

\begin{proof}
 With Lemma~\ref{lem:prop_orth_proj}, $\|z\|_2^2= \|z-P_\Sigma^\perp(z)\|_2^2 + 2 \re \ls z-P_\Sigma^\perp(z),P_\Sigma^\perp(z) \rs + \|P_\Sigma^\perp(z)\|_2^2  = \|z-P_\Sigma^\perp(z)\|_2^2 + \|P_\Sigma^\perp(z)\|_2^2 $. As $  \|P_\Sigma^\perp(z)\|_2^2 = \|z\|_2^2-\|z-P_\Sigma^\perp(z)\|_2^2 $, we deduce that  $z\mapsto \|P_\Sigma^\perp(z)\|_2^2$ is continuous if $z\mapsto \|z-P_\Sigma^\perp(z)\|_2$  is continuous. 
 
 Let $z, y \in \bC^N$, we have by definition of the orthogonal projection, as $P_\Sigma^\perp(y) \in \Sigma$, 
 \begin{equation}
 \|z-P_\Sigma^\perp(z)\|_2-\|y-P_\Sigma^\perp(y)\|_2 \leq \|z-P_\Sigma^\perp(y)\|_2-\|y-P_\Sigma^\perp(y)\|_2. 
 \end{equation}
 With the triangle inequality, 
 \begin{equation}
 	\|z-P_\Sigma^\perp(z)\|_2-\|y-P_\Sigma^\perp(y)\|_2 \leq \|z-y \|_2+\|y-P_\Sigma^\perp(y)\|_2-\|y-P_\Sigma^\perp(y)\|_2 = \|z-y\|_2.
 \end{equation}
 We similarly show that  $	\|y-P_\Sigma^\perp(y)\|_2 -\|z-P_\Sigma^\perp(z)\|_2 \leq \|z-y\|_2$ and 
 \begin{equation}
 	|\|z-P_\Sigma^\perp(z)\|_2-\|y-P_\Sigma^\perp(y)\|_2 |\leq  \|z-y\|_2 \to_{z\to y} 0.
 \end{equation}
 which implies the continuity of  $z\mapsto \|z-P_\Sigma^\perp(z)\|_2$ and the result.
\end{proof}

We now show   the existence of the  minimizers of $R_{\beta^2}(u,z)$ (that characterizes the restricted Lipschitz constant (Lemma~\ref{lem:max_Qc})) with respect to $u \in \Sigma$.

\begin{lemma}\label{lem:exist1}
	Suppose that $\Sigma$ is a closed homogeneous proximinal set. Let $\beta^\star = \inf_{P \in \Pi_\Sigma} \beta_\Sigma(P)$.  Suppose $\beta^\star>1$.  Then for all $\beta> \beta^\star, z \in \bC^N$, with $R_{\beta^2}$ defined in Lemma~\ref{lem:max_Qc}, we have 
	\begin{equation}
		P_{\beta^2}(z) := \left(\arg \min_{u\in\Sigma} R_{\beta^2}(u ,z)\right) \neq \emptyset
	\end{equation}
	with $R_{\beta^2}(P_{\beta^2}(z),z) \leq 0$ and $\|P_\beta(z)\|_2^2 \leq \beta^2 \|z\|_2^2$. 
\end{lemma}

\begin{proof}
	Let $c = \beta^2 > (\beta^{\star})^2$.	We consider the minimization of 
	\begin{equation}
		\begin{split}
			R_{c} (u,z) & = \|u\|_2^2 - c\|z\|_2^2 +\frac{1}{c-1} \|P_\Sigma^\perp(u-cz)\|_2^2. \\
		\end{split}
	\end{equation}
	
	With Lemma~\ref{lem:exist0b},  $R_c(u,z)$ is continuous with respect to $u$. Also note  that  $\inf_{u \in \Sigma}R_c(u,z)= \inf_{u \in \Sigma} \max_{x\in\Sigma} Q_c(u,z,x) \leq 0$ (using Lemma~\ref{lem:charac_lip1} and $ c=\beta^2 > (\beta^{\star})^2$)  and that if $\|u\|_2^2 > c\|z\|_2^2$, we have $	R_{c} (u,z) >0 $.  We deduce 
	\begin{equation}
		\begin{split}
			\inf_{u\in\Sigma}R_{c} (u,z)= \inf_{u\in \Sigma, \|u\|_2^2\leq c\|z\|_2^2} R_{c} (u,z)
		\end{split}
	\end{equation}
	
	As $\Sigma$ is closed, $\Sigma \cap B(\sqrt{c}\|z\|_2) $ (where $B(r)$ is the closed $\ell^2$-ball of radius $r$) is closed and bounded. By continuity, we deduce that the minimum over  $\Sigma \cap B(\sqrt{c}\|z\|_2) $  is reached  in  $\Sigma \cap B(\sqrt{c}\|z\|_2) $ and that $P_{\beta^2}(z) \neq \emptyset$.  
	
	Also note that we necessarily have $P_{\beta^2}(z) \subset \Sigma \cap B(\sqrt{c}\|z\|_2) $  as  $	R_{c} (u,z) >0 $ for $\|u\|_2 > c\|z\|_2^2$.
\end{proof}

We now prove  the existence of optimal projections for closed homogeneous proximinal sets (such as finite union of linear subspaces or sets of matrices of rank lower than  $r$).

\begin{theorem}[Existence of optimal projections]\label{th:existence}
Suppose that $\Sigma$ is a closed homogeneous proximinal set. Let $\beta^\star := \inf_{P \in \Pi_\Sigma} \beta_\Sigma(P)$. Suppose $\beta^\star>1$. 
Then, for all $z\in\bC^N$, with $R_{(\beta^\star)^2}$ defined in Lemma~\ref{lem:max_Qc},
\begin{equation}
	P^\star(z):=	\left(\arg\min_{u\in\Sigma}R_{(\beta^\star)^2}(u ,z)\right)\neq \emptyset.
\end{equation}
and $\beta_\Sigma(P^\star) = \beta^\star$.
\end{theorem}
\begin{proof}
		
	Let $c = (\beta^\star)^2$. Consider, with Lemma~\ref{lem:exist1}, for all $\epsilon>0, z \in \bC^N$,
	\begin{equation}
		P_{c+\epsilon}(z) := \left(\arg \min_{u\in\Sigma} R_{c+\epsilon}(u,z) \right)\neq \emptyset.
	\end{equation}
 and $\|P_{c+\epsilon}(z)\|_2^2 \leq (c+\epsilon)\|z\|_2^2$. Considering $c+\epsilon \leq 4$ (as $\beta^\star \leq 2$ with Lemma~\ref{th:lip_const_orth_proj} and  the orthogonal projection is optimal if $\beta^\star=2$), we have $\|P_{c+\epsilon} (z)\|_2^2 \leq 4\|z\|_2^2$ and $\|P_\epsilon(z)\|$ is bounded independently of $\epsilon$ and $c$.  We also have $R_{c+\epsilon}(P_{c+\epsilon}(z),z) \leq 0$. As $P_{c+\epsilon} (z)$ is bounded and in $\Sigma$, there is a convergent sequence $u_n = P_{c+\epsilon_n}(z)\to_{n\to+\infty}\tilde{u} \in \Sigma$ with $\epsilon_n\to_{n\to+\infty}0$ and  
 
	\begin{equation}
		P_{c+\epsilon_n}(z) \in  \arg \min_{u\in\Sigma} R_{c+\epsilon_n}(u,z) \leq 0.
	\end{equation}
	where $R_{c+\epsilon_n}(u,z)=\|u\|_2^2 - (c+\epsilon_n)\|z\|_2^2 +\frac{1}{c+\epsilon_n-1} \|P_\Sigma^\perp(u-(c+\epsilon_n)z)\|_2^2 $.
	With Lemma~\ref{lem:exist0b}, by continuity of $\|P_\Sigma^\perp(\cdot)\|_2^2 $,  we deduce 
	
	\begin{equation}
	0\geq R_{c+\epsilon_n}(u_n,z) \to_{\epsilon_n \to 0} \|\tilde{u}\|_2^2 - c\|z\|_2^2 +\frac{1}{c-1} \|P_\Sigma^\perp(\tilde{u}-cz)\|_2^2  =  R_{c}(\tilde{u},z).
	\end{equation}
	
	We have just shown that for any $z$, there exists $P(z) = \tilde{u}$ such that 
	
	\begin{equation}
	R_{(\beta^\star)^2}(P(z) ,z)\leq 0.
	\end{equation}
	which implies that $P(z)$ is restricted $\beta^\star$-Lipschitz (existence of an optimal projection).
	
	This also shows that 	
	\begin{equation}
		\inf_{u\in\Sigma}R_{(\beta^\star)^2}(u ,z) \leq R_{(\beta^\star)^2}(P(z) ,z) \leq 0.
	\end{equation}
	
	Using the same argument as in the proof Lemma~\ref{lem:exist1}, this infimum is indeed reached and  we can define
	
	\begin{equation}
	P^\star(z):=	\arg\min_{u\in\Sigma}R_{(\beta^\star)^2}(u,z) \neq \emptyset
	\end{equation}
	such that $	Q_{(\beta^\star)^2}	(P^\star(z) ,z,) \leq 0$ and $\beta_\Sigma(P^\star)= \beta^\star$.
\end{proof}

Note that the optimal projection is constructed explicitly from the pointwise minimization (i.e. minimization for each $z$) of an explicit function $R_c$ constructed with $Q_c$ and the orthogonal projection onto $\Sigma$. Also note that the proof is valid for  closed infinite unions of linear subspaces that are proximinal, such as sets of low-rank matrices.

\subsection{An optimality result for sparse recovery with iterative hard thresholding}\label{sec:opt_iht}

We observe that iterative hard thresholding for sparse recovery fits well with the previous framework. Indeed, for sparse recovery $\Sigma =\Sigma_k \subset \bR^N$ (we place ourselves in $\bR^N$ for simplicity) the set of vectors with at most $k$ non-zero elements, $P_\Sigma^\perp(z) = \mathrm{HT}(z)$ where $\mathrm{HT}$ is the hard thresholding operator selecting  $k$ largest absolute amplitudes in $z$.

The restricted Lipschitz property of the hard thresholding operator is a direct corollary of Theorem~\ref{th:lip_const_orth_proj}. We give a tighter restricted Lipschitz constant in the following theorem.

\begin{theorem}[Restricted Lipschitz property of hard thresholding] \label{th:lip_const_ht}
 Let $\Sigma = \Sigma_k \subset \bR^N$. Then $P_\Sigma^\perp$ has the restricted Lipschitz condition w.r.t to $\Sigma$ with constant $\beta = \sqrt{\frac{3 + \sqrt{5}}{2}} \approx 1.618$ .
 
\end{theorem}

\begin{proof}
We use the characterization of the Lipschitz constant with the function $Q_c$ given in Lemma~\ref{lem:charac_lip1}.
Let $c> 1$. Let $z\in\bC^N$. We write $z_S$ the restriction of $z$ to a support $S$. Let $z_T = P_{\Sigma}^\perp(z) = \mathrm{HT}(z)$, with a support $T$ that selects $k$ largest amplitudes in $z$.

We use Lemma~\ref{lem:max_Qc}. The maximum of $Q_c(z_T,z,\cdot)$ with respect to $x$ is reached at
  \begin{equation}
  \begin{split}
   x^\star &= P_{\Sigma}^\perp \left(\frac{1}{1-c}(z_T-c z)\right) = P_\Sigma^\perp \left(z_T + \frac{c}{c-1} z_{T^\mathrm{c}}\right). \\
   \end{split}
  \end{equation}
  We define $I$ the set of (less than $k$) coordinates of $z_T + \frac{ c }{1-c}z_{T^\mathrm{c}}$  selected by $P_\Sigma^\perp$ in  $T$ and $J$ the coordinates selected in $T^\mathrm{c}$: 
  \begin{equation}
  	I = \supp \left( P_\Sigma^\perp \left(z_T + \frac{c}{c-1} z_{T^\mathrm{c}}\right) \right) \cap T; \quad J = \supp \left( P_\Sigma^\perp \left(z_T + \frac{c}{c-1} z_{T^\mathrm{c}}\right) \right) \cap T^\mathrm{c}
  \end{equation}
  
   Note that $|I|+|J|=k$. Also remark that, $[P_\Sigma^\perp \left(z_T + \frac{c}{c-1} z_{T^\mathrm{c}}\right)]_I = z_I$, $[P_\Sigma^\perp \left(z_T + \frac{c}{c-1} z_{T^\mathrm{c}}\right)]_J = \frac{c}{c-1} z_J$ and
   \begin{equation}
 	\| P_\Sigma^\perp \left(z_T + \frac{c}{c-1} z_{T^\mathrm{c}}\right)\|_2^2 = \|z_I\|_2^2 + \|\frac{c}{c-1} z_J\|_2^2
   \end{equation}
    We deduce
 \begin{equation}
 \begin{split}
 Q_c(P_\Sigma^\perp(z),z,x^\star) &= Q_c(z_T,z,x^\star)\\
 &=(1-c)\|z_T\|_2^2   -c  \|z_{T^\mathrm{c}}\|_2^2    +(c-1) \|  P_\Sigma^{\perp} \left(z_T + \frac{ c z_{T^\mathrm{c}}}{1-c}\right)\|_2^2 \\
  &=(1-c)\|z_T\|_2^2   -c  \|z_{T^\mathrm{c}}\|_2^2    +(c-1) \|z_I\|_2^2 + \frac{ c^2 }{c-1} \|z_{J}\|_2^2 \\
   &=(1-c)\|z_{T\setminus I}\|_2^2   -c  \|z_{T^\mathrm{c}}\|_2^2  + \frac{ c^2 }{c-1} \|z_{J}\|_2^2 \\
  \end{split}
 \end{equation}
We remark that $|T \setminus I| = k - |I| = |J|$. By definition of $T$ any coordinate of $z_{T\setminus I}$ is larger than a coordinate of $z_J$ as $J\subset T^\mathrm{c}$. We deduce $\|z_{J}\|_2^2 \leq \|z_{T\setminus I}\|_2^2$ and, as $1-c < 0$ and $\|z_{T^\mathrm{c}}\|_2^2\geq \|z_J\|_2^2$,
   \begin{equation}\label{eq:interm_eq_th_it_lip}
 \begin{split}
 Q_c(P_\Sigma^\perp(z),z,x^\star) &\leq (1-c)\|z_{J}\|_2^2   -c  \|z_{J}\|_2^2  + \frac{ c^2 }{c-1} \|z_{J}\|_2^2 \\
      &=  \frac{ -c^2+3c-1 }{c-1} \|z_{J}\|_2^2 \\
 \end{split}
 \end{equation}

 The last expression is zeroed for $ c = c^\star := 3/2 + \sqrt{5}/2 \approx 2.62$, indeed $-(3/2 + \sqrt{5}/2 )^2 +  9/2 + 3\sqrt{5}/2 -1 =  9/4-1  -5/4 -3 \sqrt{5}/2 + 3\sqrt{5}/2 = 0$. We deduce for any $z \in \bC^N, x \in \Sigma$

\begin{equation}
\begin{split}
 Q_{c^\star}(P_\Sigma^\perp(z),z,x) &\leq 0.\\
 \end{split}
 \end{equation}
 and $P_\Sigma^\perp$ is restricted $\beta$ - Lipschitz with $\beta = \sqrt{c^\star}$.
\end{proof}

Going back to our condition we have  $  \delta \beta< 1  $  i.e. linear recovery with rate  $r = \delta^2\beta^2$ provided the operator $A$ has a RIC on $\Sigma -\Sigma = \Sigma_{2k}$ with constant
\begin{equation}
\delta < \frac{1}{\beta}= \frac{1}{\sqrt{c^\star}} \approx 0.618
\end{equation}
(the threshold for recovery with convex methods and potentially sublinear rates is $\delta < \frac{1}{\sqrt{2}}\approx 0.707$ \cite{traonmilin2018stable}). In \cite{foucart2011hard}, hard thresholding pursuit (with a linear rate of convergence) is successful if $\delta_{3k} < 0.57$. As the RIP on $\Sigma_{3k}$ is much more stringent than the RIP on $\Sigma_{2k}$, we conclude that our general result is state-of-the-art for sparse recovery.

In the following, we show that the restricted Lipschitz constant given by Theorem~\ref{th:lip_const_ht} is tight when considering the collection of all sparse models. We also show  that the orthogonal projection is a projection having the best Lipschitz constant if we consider the whole collection of sparse models. Also, for a fixed sparse model, the orthogonal projection is near-optimal.

We begin by the following technical Lemma.
\begin{lemma} \label{lem:tech_lem1}
  Let an integer $k \geq 3$ and $c \in [1,+\infty)$. Consider $\Sigma = \Sigma_k \subset \bR^N$. Let $z \in \bC^N$ defined by $z_i = 1$ for $1\leq i\leq k+1$ and $z_i = 0$ for $i> k+1$. For $u \in \Sigma$, with $R_c$ defined in Lemma~\ref{lem:max_Qc}, we have
	\begin{equation}\label{eq:res_tech_lem1}
		\min_{u\in\Sigma} R_c(u,z) = F_k\left(\frac{(k-1)c}{(k-1)c +c-1}\right)
	\end{equation} 
	where we define, for $v\in \bR$, 
	\begin{equation}\label{def:F_k_hyp}
		F_k(v):= k|v|^2 +\frac{1}{c-1} (c^2 + (k-1) |v-c|^2) - c(k+1).
	\end{equation} 
\end{lemma}

\begin{proof}
We first show that $R_c(\cdot,z)$ is minimized by $u \in \Sigma$ such that $\supp(u) \subset  \{1, \ldots, k+1\} $.

\noindent\textbf{Restriction of the support of $u^\star $.} Let $u \in \Sigma$. Let  $I = \supp(u) \cap  \{1, \ldots, k+1\} $ and $J =\supp(u) \setminus I$.
We have
\begin{equation}
	\begin{split}
		&R_c(u_I,z)-R_c(u,z) \\
		 &= \|u_I\|_2^2  + \frac{1}{c-1} \|P_\Sigma^\perp(u_I-cz)\|_2^2
		- \|u_I+u_J\|_2^2 - \frac{1}{c-1} \|P_\Sigma^\perp(u_I +u_J-cz)\|_2^2 \\
		&= -\|u_J\|_2^2  + \frac{1}{c-1} \|P_\Sigma^\perp(u_I-cz)\|_2^2 - \frac{1}{c-1} \|P_\Sigma^\perp(u_I +u_J-cz)\|_2^2. \\
	\end{split}
\end{equation}

As $J \cap \supp(u_I -cz) = \emptyset$ (as $J \cap \{1, \ldots, k+1\} =\emptyset $), we have that $\|P_\Sigma^\perp(u_I-cz)\|_2^2 = \|\mathrm{HT}(u_I-cz)\|_2^2 \leq\| \mathrm{HT}(u_I+u_J-cz)\|_2^2 =\|P_\Sigma^\perp(u_I+u_J-cz)\|_2^2 $.
We deduce that

\begin{equation}
	\begin{split}
		R_c(u_I,z)-R_c(u,z)
		&\leq -\|u_J\|_2^2 \leq 0. \\
	\end{split}
\end{equation}

We deduce that we can consider $u^\star $ minimizing $R_c(u,z)$ such that $|\supp(u^\star )|  \subset \{1, \ldots, k+1\}$, i.e
\begin{equation}
	\min_{u \in \Sigma}R_c(u,z) = \min_{u \in \Sigma, \supp(u) \subset \supp(z)} R_c(u,z).
\end{equation}

\noindent\textbf{Explicit minimization of $R_c(\cdot,z)$.} Let $u \in \Sigma$ such that $\supp(u) \subset \{1, \ldots, k+1\}$, e.g. (without loss of generality) $\supp(u) \subset \{1, \ldots, k\}$ as $z$ is constant on $\{1, \ldots, k+1\}$.

We distinguish two cases:
\begin{itemize}
	\item \textbf{Case 1:} for all $i \in \{1, \ldots, k\}, |u_i - cz_i| = |u_i-c| \geq c$. Then, $P_\Sigma^\perp(u-cz) = \mathrm{HT}(u-cz) = (u_1-cz_1,\ldots,u_k-c z_k, 0  \ldots 0 ) $ and we have
	\begin{equation} \label{eq:opt_proj0}
		\begin{split}
			R_c(u,z) &=  \sum_{i=1}^k |u_i|^2  - c(k+1)+\frac{1}{c-1} \sum_{i=1}^k |u_i-c|^2 \geq  k\frac{c^2}{c-1} - (k+1)c.\\
		\end{split}
	\end{equation}

	\item \textbf{Case 2:}  there exists $j \in \{1, \ldots, k\}$ such that $|u_j - c| < c$. Without loss of generality (just reorder the supports), suppose $|u_k - c| = \min_{ i \in \{1, \ldots, k\}} |u_i -c| < c = |u_{k+1}-c z_{k+1}| $ (as $u_{k+1} =0$ and $z_{k+1}=1$ ). Hence, $P_\Sigma^\perp(u-cz)= \mathrm{HT}(u-cz) $ necessarily selects index $k+1$. We deduce that $\|P_\Sigma^\perp(u-cz)\|_2^2=c^2 + \sum_{i=1}^{k-1} |u_i-c|^2$.	Then
	\begin{equation}\label{eq:proof_opt_ht1}
		\begin{split}
			R_c(u,z) &=  \sum_{i=1}^k |u_i|^2- c(k+1) +\frac{1}{c-1} (c^2 + \sum_{i=1}^{k-1} |u_i-c|^2) .\\
		\end{split}
	\end{equation}
	
	We minimize $R_c(u,z)$ with respect to $u_i$, $i\neq k$ and the constraint $|u_i -c| \geq |u_k-c| =: \lambda$. We have
	\begin{equation}\label{eq:proof_opt_ht2}
		\begin{split}
			\arg \min_{u_i, |u_i -c| \geq \lambda }R_c(u,z) =  \arg \min_{u_i, |u_i -c| \geq \lambda } |u_i|^2 +\frac{1}{c-1} |u_i-c|^2 =: g(u_i).
		\end{split}
	\end{equation}
	remarking that $g'(u_i) = 2u_i + \frac{2}{c-1} (u_i-c) = 0$ for $u_i=1$.  We distinguish two sub-cases:
\begin{enumerate}
	\item If $|c-1| \geq \lambda =|u_k-c|$ (the global minimum of $g$ is within the constraint), this second-degree polynomial is minimized  at $u_i^\star =1$. In this case, this gives
	\begin{equation}
		\begin{split}
			&R_c( (u_1^\star , \ldots, u_{k-1}^\star ,u_k,0 \ldots,0) ,z)\\
			 &=  (k-1)  + |u_k|^2  +\frac{1}{c-1} (c^2 + (k-1)(1-c)^2 )- c(k+1) \\
		\end{split}
	\end{equation}
	As $c\geq 1$, minimizing $|u_k|^2$ with respect to $u_k$ such that $|u_k-c| \leq |u_i^\star -c| = |c-1|$ (i.e. $1  \leq u_k \leq 2c-1$) yields $u_k^\star =1$ (as $|\cdot|^2$ is increasing for positive reals). We deduce that
	
	\begin{equation}\label{eq:opt_proj1}
		\begin{split}
			R_c(u^\star ,z) &=  k  +\frac{1}{c-1} (c^2 + (k-1)(c-1)^2 )- c(k+1)\\
			&= (c-1)(k-1) +k  -c(k+1) +\frac{c^2}{c-1}\\
			&= (c-1)k -(c-1) +(1-c)k -c+  \frac{c^2}{c-1} =  \frac{c^2}{c-1}  -2c +1;\\
		\end{split}
	\end{equation}
	i.e., using the definition of  $F_k$  in the hypotheses~\eqref{def:F_k_hyp} and the first line of~\eqref{eq:opt_proj1},
	\begin{equation}\label{eq:opt_proj3}
		\begin{split}
			\min_{u \in \bR^N: \supp(u) \subset \{1,\ldots,k \}\min_{ i \in \{1, \ldots, k\}} |u_i -c| \leq c-1  }R_c(u,z)  &= F_k(1)\\ 
			&= \frac{c^2}{c-1}  -2c +1.\\
		\end{split}
	\end{equation}
 \item 	If   $|u_k-c| = \lambda > |c-1| $,  we write $H_c(u_1, \ldots, u_{k-1}):=\sum_{i=1}^{k-1} |u_i|^2- c(k+1) +\frac{1}{c-1} (c^2 + \sum_{i=1}^{k-1} |u_i-c|^2)$ and
 \begin{equation}\label{eq:interm1}
 	\begin{split}
	&\min_{c>|u_k-c|  > |c-1|,|u_i-c|\geq |u_k-c|, i \neq k }R_c( (u_1, \ldots, u_{k-1} ,u_k,0 \ldots,0) ,z) \\
	&= \min_{c>|u_k-c|> |c-1|} \left(|u_k|^2 + \min_{ |u_i-c|\geq |u_k-c|, i \neq k}  H_c(u_1, \ldots, u_{k-1})\right)\\
		&= \min_{ 0<v<1 \;\mathrm{or} \; 2c-1<v<2c} \left(|v|^2 + \min_{ |u_i-c|\geq |v-c|, i \neq k}  H_c(u_1, \ldots, u_{k-1})\right).\\
	\end{split}
 \end{equation}
 We remark with the change of variable $ \tilde{v} = 2c-v$ and that $|c-\tilde{v}| = |c-2c+v| = |v-c|$ that 
 \begin{equation}\label{eq:interm2}
 	\begin{split}
 	& \min_{ 2c-1<v<2c} \left(|v|^2 + \min_{ |u_i-c|\geq |v-c|, i \neq k}  H_c(u_1, \ldots, u_{k-1})\right)\\
 	 &=  \min_{0<\tilde{v}<1} \left(|2c-\tilde{v}|^2 + \min_{ |u_i-c|\geq |\tilde{v}-c|, i \neq k}  H_c(u_1, \ldots, u_{k-1})\right)\\
 	 &\geq  \min_{0<\tilde{v}<1} \left(|\tilde{v}|^2 + \min_{ |u_i-c|\geq |\tilde{v}-c|, i \neq k}  H_c(u_1, \ldots, u_{k-1})\right)\\
 	\end{split}
 \end{equation}
 where the last inequality comes from the fact that  $|2c-\tilde{v}|^2 = 2c-\tilde{v} \geq 2c-1 \geq 1 \geq \tilde{v}$ as $c\geq 1$. Combining~\eqref{eq:interm1} and \eqref{eq:interm2} we deduce 
  \begin{equation}
 	\begin{split}
 		&\min_{c>|u_k-c|  > |c-1|,|u_i-c|\geq |u_k-c|, i \neq k }R_c( (u_1, \ldots, u_{k-1} ,u_k,0 \ldots,0) ,z) \\
 	=	&\min_{ 0<u_k<1,|u_i-c|\geq |u_k-c|, i \neq k }R_c( (u_1, \ldots, u_{k-1} ,u_k,0 \ldots,0) ,z) \\
 	\end{split}
 \end{equation}

 Now supposing $0<u_k<1$, starting from~\eqref{eq:proof_opt_ht1}, we minimize $R_c$ with respect to $u_i$, $i\neq k$. This yields  $u_i^\star =u_k$ (as the function $g$ from~\eqref{eq:proof_opt_ht2} is minimized on the constraint) and
	\begin{equation}\label{eq:opt_proj0b}
		\begin{split}
			&R_c((u_1^\star , \ldots, u_{k-1}^\star ,u_k,0,\ldots,0),z)\\
			 &=    k|u_k|^2 +\frac{1}{c-1} (c^2 + (k-1) |u_k-c|^2) - c(k+1) = F_k(u_k)\\
		\end{split}
	\end{equation}
	 where $F_k$ is defined in the hypotheses~\eqref{def:F_k_hyp}. We have $F_k'(u_k) = 2k u_k+ 2\frac{k-1}{c-1}(u_k-c)$ and $F_k'(u_k^\star)=0$ if  $(k(c-1) +(k-1) )u_k^\star = (k-1)c $, i.e. if $(kc-1)u_k^\star=(k-1)c$ and $u_k^\star  = \frac{(k-1)c}{kc-1} =  \frac{(k-1)c}{(k-1)c +c-1} \leq 1$. We deduce that the inequality $|c-u_k^\star |  \leq c$ is verified and
	
	\begin{equation}
		\begin{split}
			R_c(u^\star ,z) &=   F_k\left(\frac{(k-1)c}{(k-1)c +c-1}\right),  \\
		\end{split}
	\end{equation}
	i.e.
	\begin{equation}\label{eq:opt_proj2}
		\begin{split}
			\min_{u: \supp(u) \subset \{1,\ldots,k \}, c > \min_{ i \in \{1, \ldots, k\}} |u_i -c| \geq c-1  }R_c(u,z)  &=   F_k\left(\frac{(k-1)c}{(k-1)c +c-1}\right).  \\
		\end{split}
	\end{equation}
\end{enumerate}
\end{itemize}

With \eqref{eq:opt_proj2} and \eqref{eq:opt_proj3} and using the fact that $F_k\left(\frac{(k-1)c}{(k-1)c +c-1}\right)$ is the global minimum of $F_k(\cdot)$, we get

\begin{equation}
	\begin{split}
		\min_{u: \supp(u) \subset \{1,\ldots,k \}, \min_{ i \in \{1, \ldots, k\}} |u_i -c| <c  }R_c(u,z)  &= \min \left( F_k(1), F_k\left(\frac{(k-1)c}{(k-1)c +c-1}\right)\right)  \\
		&=  F_k\left(\frac{(k-1)c}{(k-1)c +c-1} \right).\\
	\end{split}
\end{equation}

We compare with the lower bound from \eqref{eq:opt_proj0}, we have, for $k>2$

\begin{equation}
	\begin{split}
		k\frac{c^2}{c-1} - (k+1)c - F_k(1) &= (k-1)\frac{c^2}{c-1} - (k-1)c -1\\
		&= (k-1)c(\frac{c}{c-1} - 1) -1\\
		&= (k-1)\frac{c}{c-1} -1 \geq k-2 >0.
	\end{split}
\end{equation}
We deduce that for $k\geq 3$

\begin{equation}
	\begin{split}
		\min_{u: \supp(u) \subset \{1,\ldots,k \}} R_c(u,z)  &= \min_{u: \supp(u) \subset \{1,\ldots,k \}, \min_{ i \in \{1, \ldots, k\}} |u_i -c| <c  }R_c(u,z) \\
		&=  F_k\left(\frac{(k-1)c}{(k-1)c +c-1}\right).\\
	\end{split}
\end{equation}	
\end{proof}
Note that this Lemma shows that the projection optimizing $R_c$ is not the orthogonal projection for such choices of $z$. Hence the question of optimal restricted $\beta$-Lipschitz projection for fixed $k$ is still open.

\begin{theorem}[An optimality result for hard thresholding] \label{th:optimality_orth_union}
For any $k \geq 1$, consider  $\Sigma =\Sigma_k \subset \bR^N$ the $k$-sparse model set. Let $\Pi_{\Sigma_k}$ the set of projections onto $\Sigma_k$ with a restricted Lipschitz property. Let
\begin{equation}
\beta_k^\star = \inf_{P \in \Pi_{\Sigma_k}} \beta_{\Sigma_k}(P).
\end{equation}
Then,
\begin{enumerate}
 \item  the restricted Lipschitz constant from Theorem~\ref{th:lip_const_ht} is tight when considering the collection of sparse models for all sparsities $k \geq 1$, i.e.
\begin{equation}
\sup_{k\geq1} \beta_{\Sigma_k}^2(P_{\Sigma_k}^\perp)= \frac{3 +\sqrt{5}}{2} ;
\end{equation}

 \item for any $k \geq 3$,

\begin{equation}
2.457 \approx
   \frac{3 +\sqrt{11/3}}{2}   \leq (\beta_k^\star)^2 \leq \beta_{\Sigma_k}^2(P_{\Sigma_k}^\perp)\leq \frac{3 +\sqrt{5}}{2}   \approx 2.618.
\end{equation}

 \item we have the optimality of the orthogonal projection with respect to the sequence of models $\Sigma_{k}$:
\begin{equation}
   \sup_{k \geq 3} (\beta_k^\star )^2=  \sup_{k\geq 3}\beta_{\Sigma_k}^2(P_{\Sigma_k}^\perp) = \frac{3 +\sqrt{5}}{2}.
\end{equation}
\end{enumerate}

\end{theorem}

\begin{proof}

Recall that the study of $R_c$ from from Lemma~\ref{lem:max_Qc}, permits to characterize restricted Lipschitz constants.

The idea of this proof is to exactly optimize $R_c$  for a specific choice  of $z$, yielding necessary conditions on $c$ (a lower bound) to obtain $R_c(u^\star,z)= Q_c(u^\star,z,x^\star(z)) \leq 0$.  Indeed, we have  
\begin{equation}
	\sup_{\tilde{z} \in \bC^N}  \inf_{u \in \Sigma} R_c(u,\tilde{z}) \geq   \inf_{u \in \Sigma} R_c(u,z) .
\end{equation}
 Hence,  $\inf_{u \in \Sigma} R_{\beta_0^2}(u,z) > 0$ implies $	\sup_{\tilde{z} \in \bC^N}  \inf_{u \in \Sigma} R_{\beta_0^2}(u,\tilde{z}) > 0$ and the optimal projection on $\Sigma$ has  restricted Lipschitz constant $\beta> \beta_0$ (with Theorem~\ref{th:existence}).
 
Let $z \in \bC^N$ defined by $z_i = 1$ for $1\leq i\leq k+1$ and $z_i = 0$ for $i> k+1$. With Lemma~\ref{lem:tech_lem1},  we have, for $k \geq 3$,
\begin{equation}\label{eq:interm_result1}
 \min_{u\in\Sigma} R_c(u,z) = F_k\left(\frac{(k-1)c}{(k-1)c +c-1}\right)
\end{equation} 
 where we define, for $v\in \bR$, 
 \begin{equation} \label{eq:def_F_k_th}
 F_k(v)= k|v|^2 +\frac{1}{c-1} (c^2 + (k-1) |v-c|^2) - c(k+1).
\end{equation} 

 With this result, we will show our conclusion by looking at $k\to\infty$ (for  item 1 and 3) and the case $k=3$ (lower bound in item 2.).

\noindent\textbf{Best uniform constant independent of $k$.} Keeping the same $z$, we have shown that

\begin{equation}
\begin{split}
 \sup_{\tilde{z} \in \bC^N} \inf_{u\in\Sigma} R_c(u,\tilde{z}) \geq  \min_{u \in \Sigma} R_c(u,z) = F_k\left(\frac{(k-1)c}{(k-1)c +c-1}\right).\\
 \end{split}
\end{equation}

We now show $F_k\left(\frac{(k-1)c}{(k-1)c +c-1}\right) \to_{k\to+\infty} F_k(1)= \frac{c^2}{c-1}  -2c +1$. Let $w_k :=\frac{(k-1)c}{(k-1)c +c-1} = \frac{1}{1+v_k}$ with $v_k:=\frac{c-1}{(k-1)c}$. With the definition of $F_k$ in~\eqref{eq:def_F_k_th},
\begin{equation}\label{eq:conv_F_cond_nec1}
\begin{split}
	F_k(w_k) &= kw_k^2+\frac{k-1}{c-1}(w_k-c)^2 +\frac{c^2}{c-1} -c(k+1) \\
	&=[ kw_k^2+\frac{k}{c-1}(w_k-c)^2  -ck]  +[\frac{c^2}{c-1} - \frac{1}{c-1}(w_k-c)^2  -c] \\
\end{split}
\end{equation}
As $w_k \to_{k\to+\infty} 1$, we have the second term $[\frac{c^2}{c-1} - \frac{1}{c-1}(w_k-c)^2  -c] \to_{k\to +\infty}\frac{c^2}{c-1} -(c-1) -c = \frac{c^2}{c-1} -2c +1 = F_k(1)$ . We rewrite the first term of~\eqref{eq:conv_F_cond_nec1} as  
\begin{equation}
	\begin{split}
		kw_k^2+\frac{k}{c-1}(w_k-c)^2  -ck &= \frac{k}{c-1}((c-1)w_k^2 +(w_k-c)^2 -c) \\
		&= \frac{k}{c-1}(cw_k^2 -2w_kc +c^2-c(c-1))\\
		&= \frac{ck}{c-1}(w_k^2 -2w_k+1)\\
		&= \frac{ck}{c-1}(w_k-1)^2 = \frac{ck}{c-1}(\frac{1}{1+v_k}-1)^2\\
		& = \frac{ck}{c-1}(\frac{v_k}{1+v_k})^2 \to_{k \to +\infty} 0 \\
	\end{split}
\end{equation}
where the convergence is deduced from the fact that $v_k = O(1/k)$.
We deduce from~\eqref{eq:conv_F_cond_nec1} that $F_k\left(\frac{(k-1)c}{(k-1)c +c-1}\right) \to_{k\to+\infty} F_k(1)$.

As $F_k(1) =  \frac{-c^2+3c-1}{c-1}$, for $c\geq 1$, $F_k(1) \leq 0$ if and only if $c \geq \frac{3 +\sqrt{5}}{2} \approx 2.6180$ (as shown in Theorem~\ref{th:lip_const_ht} after \eqref{eq:interm_eq_th_it_lip}).

We have  that $P_\Sigma^\star \in \arg \min_{P \in \Pi_\Sigma} \beta_\Sigma(P) $  implies

\begin{equation}
\begin{split}
 \sup_{\tilde{z} \in \bC^N}  R_c(P_\Sigma^\star(\tilde{z}),\tilde{z}) = \sup_{\tilde{z} \in \bC^N} \inf_{u\in\Sigma} R_c(u,\tilde{z}) \geq F_k\left(\frac{(k-1)c}{(k-1)c +c-1}\right).
 \end{split}
\end{equation}

We deduce that 
\begin{equation}
	\begin{split}
		 \sup_{k\geq3} \sup_{\tilde{z} \in \bC^N}  R_c(P_\Sigma^\star(\tilde{z}),\tilde{z}) \geq F_k \left(1\right) = \frac{-c^2+3c-1}{c-1}.
	\end{split}
\end{equation}
and $\sup_{k\geq3} \sup_{\tilde{z} \in \bC^N}  R_c(P_\Sigma^\star(\tilde{z}),\tilde{z}) \leq 0$ implies $c\geq \frac{3 +\sqrt{5}}{2} $ and $\sup_{k\geq 3} \beta_{\Sigma_k}(P_{\Sigma_k}^\star)^2 \geq \frac{3 +\sqrt{5}}{2} $. Using Theorem~\ref{th:lip_const_ht}, we have that

\begin{equation}
\begin{split}
  \frac{3 +\sqrt{5}}{2} &\geq \sup_{k\geq 1} \beta_{\Sigma_k}(P_{\Sigma_k}^\perp)^2 \geq \sup_{k\geq 3} \beta_{\Sigma_k}(P_{\Sigma_k}^\star)^2 \geq  \frac{3 +\sqrt{5}}{2}.\\
\end{split}
\end{equation} 
This proves the third conclusion of this theorem. Note that this also forces the equality $\sup_{k\geq 1} \beta_{\Sigma_k}(P_{\Sigma_k}^\perp)^2 = \frac{3 +\sqrt{5}}{2}$, which is the first conclusion of this theorem.

\noindent\textbf{Lower bound with $k=3$.} We give the lower bound in item 2 of the conclusion by simply looking at the case $k=3$. Indeed for such $k$, we have, using the definition of $F_k$, from~\eqref{eq:opt_proj0b},
\begin{equation}
\begin{split}
 &F_k\left(\frac{(k-1)c}{(k-1)c +c-1}\right) = F_k\left(\frac{2c}{3c-1}\right) \\
 &=  3\left(\frac{2c}{3c-1}\right)^2 +\frac{1}{c-1} (c^2 +2 |\frac{2c}{3c-1}-c|^2 )- 4c \\
&=  c\left(3 c\left(\frac{2}{3c-1}\right)^2 +\frac{c}{c-1} (1 +2 |\frac{2}{3c-1}-1|^2 )- 4\right)\\
&=  c\left( 3c\left(\frac{2}{3c-1}\right)^2 +\frac{c}{c-1} (1 +2 \left(\frac{-3c +3}{3c-1}\right)^2 )- 4\right)\\
\end{split}
\end{equation}
We reduce  to the same denominator.
\begin{equation}
\begin{split}
 &F_k\left(\frac{(k-1)c}{(k-1)c +c-1}\right)\\
 &=c\left(  \frac{12c(c-1) +c (3c-1)^2 +2c (-3c+3)^2 - 4 (3c-1)^2(c-1)} {(3c-1)^2(c-1)} \right)\\
&=c\left(  \frac{12c(c-1) +c (3c-1)^2 +18c (c-1)^2 - 4 (3c-1)^2(c-1)} {(3c-1)^2(c-1)} \right)\\
&=c\left(  \frac{(c-1) (12c + 18c(c-1)) +c (3c-1)^2  - 4 (3c-1)^2(c-1)} {(3c-1)^2(c-1)} \right)\\
&=c\left(  \frac{(c-1) (18c^2-6c) +c (3c-1)^2  - 4 (3c-1)^2(c-1)} {(3c-1)^2(c-1)} \right)\\
&=c\left(  \frac{6c(c-1) (3c-1) +c (3c-1)^2  - 4 (3c-1)^2(c-1)} {(3c-1)^2(c-1)} \right)\\
&=c\left(  \frac{6c(c-1) +c (3c-1)  - 4 (3c-1)(c-1)} {(3c-1)(c-1)} \right)= c \frac{-3c^2 + 9c -4}{(3c-1)(c-1)}.\\
 \end{split}
\end{equation}
Calculating the roots of the second-degree polynomial in $c$ given by $-3c^2 + 9c -4$, we deduce that $F_k\left(\frac{(k-1)c}{(k-1)c +c-1}\right)\leq 0$ for  $c\geq \frac{3 +\sqrt{11/3}}{2} \approx 2.457$.

Hence, for $k \geq 3$,

\begin{equation}
  2.457   \leq \beta_{\Sigma_k}^2(P_{\Sigma_k}^\star) \leq \beta_{\Sigma_k}^2(P_{\Sigma_k}^\perp)^2   \approx 2.618.
\end{equation}
This is the second conclusion of this Theorem.
\end{proof}

\subsection{Discussion}

We have shown that for generic models (homogeneous sets), the orthogonal projection plays an important role within the set of possible projections onto $\Sigma$. In particular, it is nearly optimal for sparse recovery.  More surprisingly, for a fixed sparsity (or an arbitrary union of subspaces), the orthogonal projection might not be optimal. Our investigations reveal that it minimizes $Q_c(u,z,x^\star)$ only for  some $z \in \bC^N$ (but not all as shown in Lemma~\ref{lem:tech_lem1}). However the bound on Lipschitz constants shows that there is little to be gained with another projection in the case of sparse recovery. We also expect these results to extend naturally to low-rank matrix recovery with iterative singular value thresholding.

Another important conclusion that we can draw from these results, is that for general sets, we can always bound reasonably the restricted Lipschitz constant of the orthogonal projection. Finding optimal projections for a general model is still an open question. In particular, does the orthogonal projection have the same near-optimality result for an arbitrary union of subspaces as iterative hard thresholding for sparse recovery?

The notion of optimality developed in this article relates to the exact recovery of elements of the model set $\Sigma$. In the presence of different sources of noise (measurement noise, modeling error, projection error), results posterior to this article show that additional stability terms are added to the convergence guarantee~\eqref{eq:res_linear_rate} from Theorem~\ref{th:gen_convevergence} \cite{joundi2025sparse}. The analysis of stability  of recovery shows that trade-offs between identifiability and stability to noise can be made if there is some knowledge on the structure of the noise. In such context the question of designing meaningful optimality criteria remains largely open. However, if the recovery method is blind to the structure of measurement noise, the obtained constants  are improving when the restricted Lipschitz constant of the projection is decreasing. Hence, the restricted Lipschitz constant as an optimality criterion stays meaningful in this (frequent) scenario.

\section{Towards the design of optimal methods of averaged directions}\label{sec:av_direc}
In this Section, we discuss how GPGD can be seen in the broader class of averaged directions method from a design and optimality perspective. We define the class of averaged directions methods as the iterations:

\begin{equation}
	\begin{split}
		x_{n+1} &= x_n - \mu d_n \\
	\end{split}
\end{equation}
where we decompose
\begin{equation}
	\begin{split}
		d_n = A^H(Ax_n-y) +g_n \\
	\end{split}
\end{equation}
where $g_n = g(x_n)$ is a function of $x_n$ only. At each step, a direction is calculated as the average of a data-fit direction and a regularizing direction. Such algorithms can be viewed as the result of a design that does not rely on an underlying functional to minimize: the iterations are built by  averaging:
\begin{itemize}
	\item a back-projection (a way to map the measurements to the ambient space $\bC^N$) of the residual between the measurements $y$ and the current iterate;
	\item and a direction $g_n$ pushing towards the low-dimensional model $\Sigma$ (that might not be the gradient of a regularization function).
\end{itemize}
We can model both convex and ``non-convex'' approaches for sparse recovery and beyond (even in the case where there is no functional to minimize).

\begin{itemize}
	\item In the convex case, we can set $g_n =\lambda \nabla R(x_n)$. We fall exactly on the gradient descent for~\eqref{eq:minimization1}.
	
	\item To obtain generalized projected gradient descent, we can choose $g_n  = ( I-A^HA)(x_n-P_\Sigma(x_n))$ where $P_\Sigma$ is a generalized projection onto the model set (see Section~\ref{sec:choice_reg_dir}). 
	
\end{itemize}

\subsection{Design of the data-fit direction $h_n$}\label{sec:choice_datafit}

To design an iterative algorithm, it is natural to first choose a direction that pushes towards $\gt$ from $x_n$, i.e. an approximation of the direction $x_n-\gt$. As  we have access to $r_n = Ax_n-y = A(x_n-\gt)$, the most common direction that helps convergence to an estimation of $\gt$ is the back-projection of the residual
\begin{equation}
	h_n = A^H(Ax_n-y), 
\end{equation}
which matches the gradient of an $\ell^2$ data-fit functional. We call $h_n$ the data-fit direction.
If $x_n$ is close to $\Sigma$, under a restricted isometry hypothesis on $A$ with small RIC constant $\delta$ (Definition~\ref{def:RIC}), we have that $A^Hr_n = A^HA(x_n-\gt) \approx x_n-\gt$ (because if $z = x_1-x_2$ with $x_1,x_2 \in \Sigma$, $\|z -A^HAz\|_2 \leq \delta \|z\|_2$) . Also, if $g_n = 0$ we fall on the Landweber iteration: $d_n=h_n =  \nabla F(x_n)$ with $F(x) = \frac{1}{2}\|Ax_n-y\|_2^2$ (gradient descent of under-determined least-squares functional). In this article, we use this specific choice of $h_n$ which can be interpreted in a purely geometrical way without considering $F$: we just use the adjoint of the measurement operator. Other possibilities could be the (sub)-gradient of other data-fit functionals such as the $\ell^1$ data-fit for robust regression. We can also mention implicit schemes such as the forward-backward algorithm where the gradient direction is estimated at $x_{n+1}$ instead of $x_n$. Such generalizations are left for future work as the choice of a data-fit direction is also linked with the type of noise  in the noisy case (which is out of the scope of this paper).

\subsection{Design of the regularizing direction $g_n$} \label{sec:choice_reg_dir}

Given the choice of data-fit directions of Section~\ref{sec:choice_datafit} and given a model set $\Sigma$, what is the best way to make $x_n$ converges towards an element of  $\Sigma$?
Using the direction $g_n = x_n - P_\Sigma^\perp(x_n)$ where $P_\Sigma^\perp$ is an orthogonal projection onto $\Sigma$ seems a good choice as it is the shortest path between $x_n$ and $\Sigma$ for the $\ell^2$ metric. Choosing   $g_n =(I/\mu-A^HA)(x_n - P_\Sigma(x_n))$ leads to GPGD iterations. Indeed, we have
\begin{equation}
	\begin{split}
		x_{n+1} &= x_n - \mu A^H(Ax_n-y) - (I-\mu A^HA)(x_n - P_\Sigma(x_n)) \\
		&= x_n  - \mu A^HAx_n +  \mu A^Hy - x_n+ P_\Sigma(x_n) + \mu A^HA(x_n - P_\Sigma(x_n))\\
		&= P_\Sigma(x_n)  - \mu A^HA( P_\Sigma(x_n)  - \gt ). \\
	\end{split}
\end{equation}

Our framework can be interpreted as a definition of optimal regularizing directions with the class of directions having the form $g_n = (I/\mu-A^HA)(x_n - P_\Sigma(x_n))$. Our results open the question of defining optimality for more general classes of regularizing directions.

\subsection{Open questions for the design of optimal averaged directions} 
A typical problem  associated to this question is to define classes of algorithms sufficiently large to yield interesting design choices, but not too large to not be of any interest. For example, to have an interesting meaning from an optimization point of view, we should consider regularizing directions as functions of $x_n$ that do not depend on $y = A \gt$. Otherwise, we could just choose  $ g_n $ depending on $\arg \min_{x\in \Sigma} \|Ax-y\|_2^2$ for a convergence in one iteration.

Another question arising from this work is to compare GPGD with more general regularizing directions.  Let us try to replicate the proof of Theorem~\ref{th:gen_convevergence} with iterations
\begin{equation}
	x_{n+1} = x_n - \mu A^H(Ax_n-y) - g(x_n)
\end{equation}
for some choice of regularizing direction $g(x_n)$ that only depend on $x_n$. Suppose for the sake of discussion that  $A^HA$ is optimally scaled and that $\mu=1$. We have 

\begin{equation}
	\|x_{n+1} -\gt\|_2 =  \| (I  - A^HA)(x_n-\gt) -  g(x_n)\|_2.
\end{equation}
If we consider that the RIC assumption  on $A^HA$ is minimal, only terms of the form $(I  - A^HA)(u-v)$ with $u,v\in\Sigma$ can be bounded by a contractive term. Setting $v=\gt$, we have, with the triangle inequality,
\begin{equation}
	\begin{split}
		\|x_{n+1} -\gt\|_2 &=  \| (I  - A^HA)(u-\gt) + (I  - A^HA)(x_n-u) -  g(x_n)\|_2\\ 
		&\leq 	\| (I  - A^HA)(u-\gt)\|_2 + \|(I  - A^HA)(x_n-u) -  g(x_n)\|_2\\ 					   
	\end{split}
\end{equation}
We immediately observe that setting $g(x_n) = (I  - A^HA)(u-x_n) $  solves two problems at the same time: it makes the second positive term $0$ (hence the triangle inequality tight) and uses only the minimal RIC assumption. We further notice that simply renaming $u= P_\Sigma(x_n)$  leads to the generalized projected gradient descent. This remark opens the broader question of the potential optimality of GPGD into broader classes methods of averaged directions  (or of  the existence of averaged direction algorithms with faster linear convergence rates than GPGD).

\section{Application to inverse problems with deep priors}\label{sec:deep_priors}

We have shown in the previous section that generalized projected gradient descent identifies low-dimensional models with a linear rate as soon as the  projection $P_\Sigma$ has the restricted Lipschitz condition (Definition~\ref{def:lip_const}) on the low-dimensional model $\Sigma$.  We now show that the plug-and-play (PnP) framework can be interpreted as a low-dimensional model recovery and that, experimentally, for image inverse problems, linear rates of convergence to the underlying low-dimensional model are observed. This shows that the restricted Lipschitz condition appears to hold approximately in practice and that global fast rates are obtained beyond the typical non-convex setting of the literature.

\subsection{An interpretation of the plug-and-play method  with low-dimensional recovery theory}

Many variations of  PnP methods exist in the literature~\cite{zhang2021plug,hurault2021gradient}. PnP methods use a general denoiser to approximate the proximal operator associated with a regularization function, which could be interpreted as an operator minimizing a distance between a low-dimensional model and a given point, i.e. a projection operator. In the following we use the formalism of the proximal gradient method \cite{liu2021recovery,kamilov2017plug} (PnP-PGM), which directly uses the denoiser in a generalized projected gradient descent scheme and  yields state-of-the-art results for inverse problems in imaging. In this context, the problem of estimating $\gt$ from $y = A \gt$ is solved using :
\begin{equation}\label{eq:PnP_it}
 x_{n+1} = D(x_n- \mu A^H(Ax_n-y))
\end{equation}
where $D$ is the general purpose denoiser and $\mu$ is the gradient step size. We fall in the iterations defined in  Theorem~\ref{th:gen_convevergence} (with projection and descent steps reversed). Hence \emph{global} convergence will be guaranteed if $D$ is a generalized projection onto a set $\Sigma$ (which is exactly the set of fixed points of the denoiser $D$) with \emph{restricted} Lipschitz constant $\beta_\Sigma(D)$. Note that, to the best of our knowledge convergence of PnP methods rely on a much more stringent \emph{global} Lipschitz condition on the considered objects (objective functional, etc,..., see  Section~\ref{sec:related_work}).  We show  experimentally that the convergence to the underlying low-dimensional model set $\Sigma$ (the fixed points of the denoiser $D$)  corresponds to  linear rates. We will illustrate this linear convergence in the following experimental sections. The main advantage of the restricted Lipschitz condition is to provide a framework to understand  simultaneously sparse recovery and such PnP methods.

Note that we opened the question of the place of GPGD compared to other regularizing directions (Section~\ref{sec:av_direc}). As PnP-PGM can be interpreted as a generalized projected gradient descent, this method might have better rates than the other natural choice of regularizing direction $(I-D)(x_n)$ (see Section~\ref{sec:choice_reg_dir}) which was indeed proposed in the gradient method regularization by denoising (GM-RED)~\cite{romano2017little} :
\begin{equation}\label{eq:PnP_RED}
 x_{n+1} = x_n- \mu (A^H(Ax_n-y) + \lambda(I-D)(x_n)).
\end{equation}
We also propose a comparison of the two methods in Section~\ref{sec:comp_PGM_RED} in light of our results. 

\subsection{Experiments}

In practice, PnP-PGM and GM-RED never fully recover $\hat{x}$ due to approximations error (that will be quantified experimentally). Indeed, the unknown $\hat{x}$ is never a perfect fixed point of the denoiser $D$ and thus we cannot expect it to be fully recovered. Further approximation errors may also occur which makes it complicated for the theory to perfectly fit. Hence, it is near impossible for the sequence $(x_{n})_{n\geq 1}$ given by \eqref{eq:PnP_it} to have a linear convergence rate with respect to $\hat{x}$ as stated by Theorem \ref{th:gen_convevergence}. However, under the assumption that $\|x_n-x^\ast\|_2^2$ has a $r$-linear rate of convergence with  $x^\ast= \lim\limits_{n\rightarrow +\infty} x_n$, we have
\begin{equation}\label{eq:numerical_linear_bound}
	\begin{split}
		\|x_n-\hat{x}\|_2^2 &\leq (\|x_n - x^\ast\|_2 + \|x^\ast - \hat{x}\|_2)^2 \\
			 &\leq ( r^{\frac{n}{2}}\|x_0 - x^\ast\|_2 + \|x^\ast - \hat{x}\|_2)^2.
	\end{split}
\end{equation}
This illustrates the fact that in case of linear convergence to an approximate $\hat{x}$ without knowing the exact limit $x^\ast$, we should observe a convergence given by the last inequality in~\eqref{eq:numerical_linear_bound}.

In the following experiments, we compare this estimation of convergence with other sublinear convergence rates $\frac{1}{n}$, $\frac{1}{n^2}$ on both synthetic images and natural images. We conduct our experiments on two different linear operators $A$: a mask operator that erases $30\%$ of the pixels of the image and a Gaussian blur operator (with $\sigma=1.0$, $3.0$). For each algorithm, the tuning parameters $\mu$ and $\lambda$ are selected through line search such that they ensure the best recovery of $\hat{x}$ (independently for each method). Moreover, the initialization $x_0$ is set using a random uniform distribution and the mean is reset such that it matches the target image mean. We verify  empirically that the restricted isometry constant and the restricted Lipschitz constant are controlled within PnP-PGM iterations (Figure \ref{fig:verification_RIP_beta}). In particular, under the consideration that $D(x_n), \hat{x}\in \Sigma$, we observe that the hypothesis of Theorem \ref{th:gen_convevergence} are verified. In particular, the measured restricted $\beta$-Lipschitz constant is bounded by ${\approx}1.6$. We remark that we do not expect the RIP to be completely satisfied as there is an incompressible error since $\hat{x}$ is an imperfect fixed-point of the denoiser $D$. We also verify the fact that the denoiser nearly realizes a projection (in the sense that it is approximately idempotent) with respect to the iterations $x_n$ (i.e. $D^2(x_n)\approx D(x_n)$). Note that these controls are verified for all experiments, for concision we only present them for the butterfly image.

The code and weights of the denoisers used for the numerical experiments can be found in the open-access GitLab repository \cite{gitlab-link}. We show in Figure~\ref{fig:original_images}, the synthetic and natural images and initializations used in our experiments.

\begin{figure}
	\centering
	\begin{subfigure}[b]{5cm}
		\centering
		\includegraphics[width=0.99\linewidth]{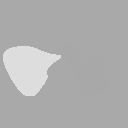}
		\caption{synthetic $\hat{x}$}
		\label{fig:cartoon_x0}
	\end{subfigure}
	\begin{subfigure}[b]{5cm}
		\centering
		\includegraphics[width=0.99\linewidth]{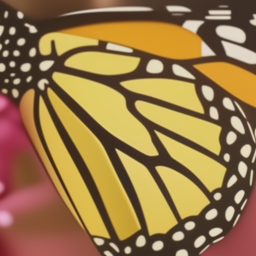}
		\caption{butterfly $\hat{x}$}
		\label{fig:butterfly_pnp_pgm_x0}
	\end{subfigure}
	\begin{subfigure}[b]{5cm}
		\centering
		\includegraphics[width=0.99\linewidth]{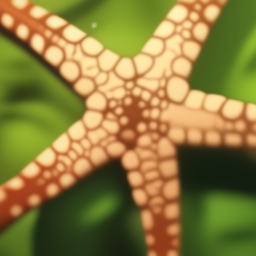}
		\caption{starfish $\hat{x}$}
		\label{fig:starfish_x0}
	\end{subfigure}
	\begin{subfigure}[b]{5cm}
		\centering
		\includegraphics[width=.99\linewidth]{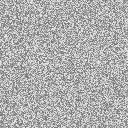}
		\caption{initialization $x_0$}
		\label{fig:cartoon_init}
	\end{subfigure}
	\begin{subfigure}[b]{5cm}
		\centering
		\includegraphics[width=.99\linewidth]{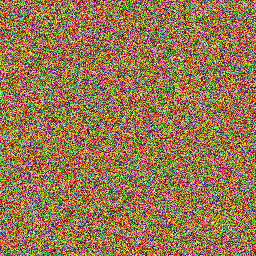}
		\caption{initialization $x_0$}
		\label{fig:butterfly_pnp_pgm_init}
	\end{subfigure}	
	\begin{subfigure}[b]{5cm}
		\centering
		\includegraphics[width=.99\linewidth]{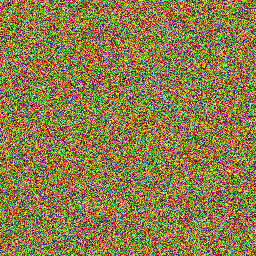}
		\caption{initialization $x_0$}
		\label{fig:starfish_init}
	\end{subfigure}
	\caption{Target images used in numerical experiments. For each numerical experiment, the image $\hat{x}$ is set such that it is approximately a fixed point of the respective denoiser. (a) The synthetic image is generated such that it is an approximate of a fixed point of the denoiser. (b-c) The DRUNet denoiser is applied a few times on a natural image  such that we obtain an approximation of a fixed point of the denoiser. (d-f) The initialization for each image respectively, generated from a random uniform distribution and with the same mean as the respective target image.}
	\label{fig:original_images}
	\vspace{-5mm}
\end{figure}

\subsubsection{Synthetic piecewise-constant images}
In this experiment, we use a  denoiser $D$ (without parameter for the  noise level), parametrized with  a DRUNet  architecture~\cite{zhang2021plug} (a state-of-the-art combination of ResNet and U-Net architectures, without explicit low-dimensional latent space), trained on a dataset of randomly generated piecewise-constant images (see~\cite{guennec2024joint} for a precise description of the random  synthetic dataset). The target image $\hat{x}$ (see Fig.~\ref{fig:cartoon_x0}) was generated such that it had a high fixed point value with respect to the denoiser ($\text{PSNR}(\hat{x}, D(\hat{x}))\approx 64.84$). We observe that the theoretical bound \eqref{eq:numerical_linear_bound} matches well with the observed convergence rate $\|x_n - \hat{x}\|$~in Figure~\ref{img:cartoon_pnp_pgm}, whereas the sublinear rates $\frac{1}{n}$ and $\frac{1}{n^2}$ did not. The theoretical rate in Figure~\ref{img:cartoon_pnp_pgm} is calculated as the minimal rate upper-bounding the experimental convergence curve.


\begin{figure}[]
	\centering
	\begin{subfigure}[b]{0.3\linewidth}
		\centering
		\includegraphics[width=0.99\linewidth]{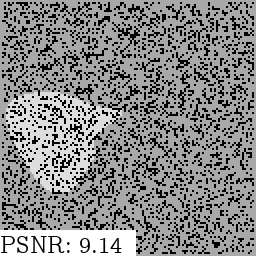}
		\caption{$y$ inpainting}
		\label{fig:cartoon_pnp_pgm_y_mask}
	\end{subfigure}
	\begin{subfigure}[b]{0.3\linewidth}
		\centering
		\includegraphics[width=0.99\linewidth]{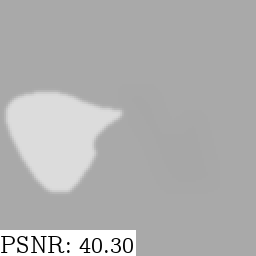}
		\caption{$y$ low blur $\sigma=1.0$}
		\label{fig:cartoon_pnp_pgm_y_blur_low}
	\end{subfigure}
	\begin{subfigure}[b]{0.3\linewidth}
		\centering
		\includegraphics[width=0.99\linewidth]{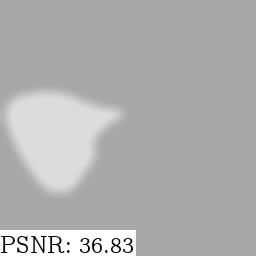}
		\caption{$y$ high blur $\sigma=3.0$}
		\label{fig:cartoon_pnp_pgm_y_blur_high}
	\end{subfigure}
	\begin{subfigure}[b]{0.3\linewidth}
		\centering
		\includegraphics[width=0.99\linewidth]{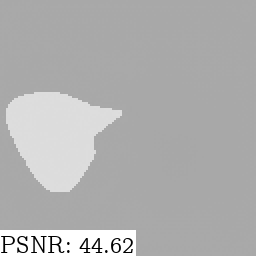}
		\caption{$x^*$ (inpainting)}
		\label{fig:cartoon_pnp_pgm_output_mask}
	\end{subfigure}
	\begin{subfigure}[b]{0.3\linewidth}
		\centering
		\includegraphics[width=0.99\linewidth]{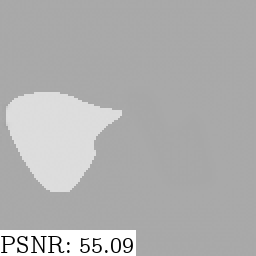}
		\caption{$x^*$ (low blur)}
		\label{fig:cartoon_pnp_pgm_output_blur_low}
	\end{subfigure}
	\begin{subfigure}[b]{0.3\linewidth}
		\centering
		\includegraphics[width=0.99\linewidth]{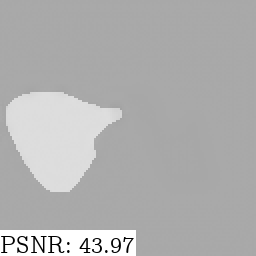}
		\caption{$x^*$ (high blur)}
		\label{fig:cartoon_pnp_pgm_output_blur_high}
	\end{subfigure}
	\begin{subfigure}[b]{0.325\linewidth}
		\centering
		\includegraphics[width=0.99\linewidth]{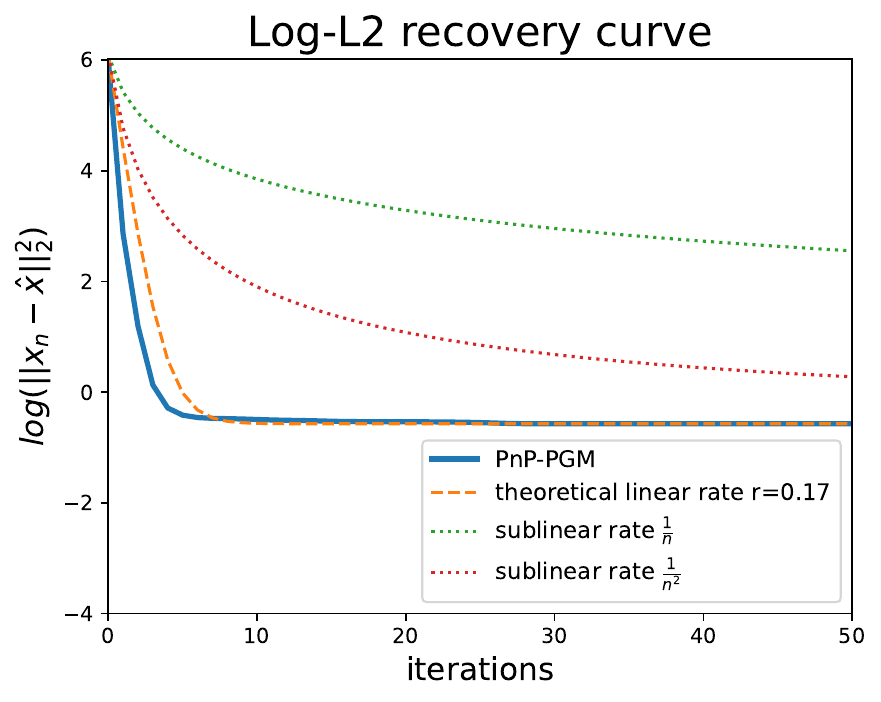}
		\caption{Log-L2 recovery curve (inpainting)}
		\label{plot:synthetic_inpainting_pnp_pgm}
	\end{subfigure}
	\begin{subfigure}[b]{0.325\linewidth}
		\centering
		\includegraphics[width=.99\linewidth]{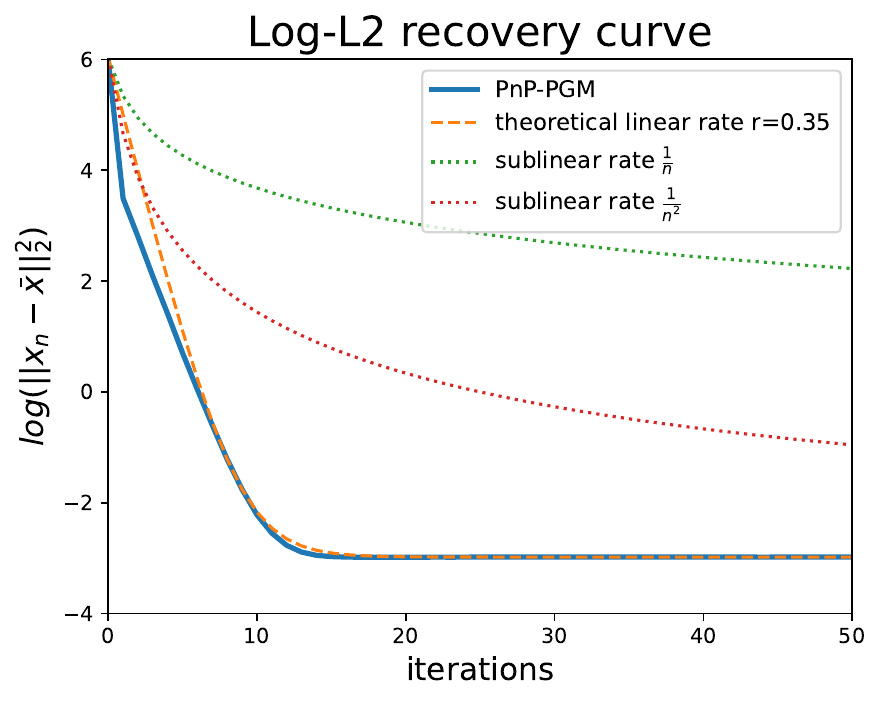}
		\caption{Log-L2 recovery curve (low blur $\sigma=1.0$)}
		\label{plot:synthetic_blur_low_pnp_pgm}
	\end{subfigure}
	\begin{subfigure}[b]{0.325\linewidth}
		\centering
		\includegraphics[width=0.99\linewidth]{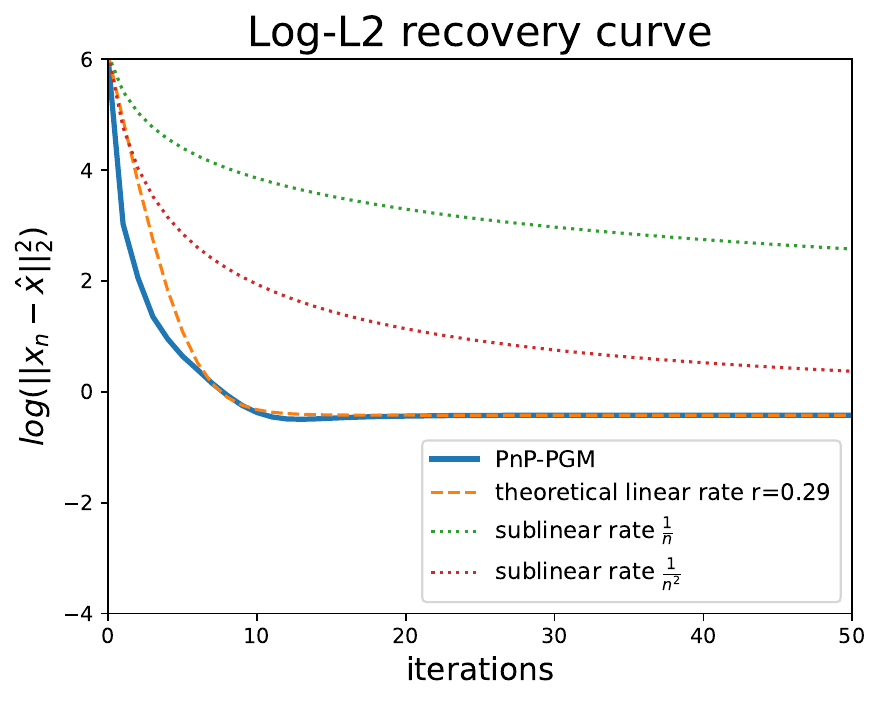}
		\caption{Log-L2 recovery curve (high blur $\sigma=3.0$)}
		\label{plot:synthetic_blur_high_pnp_pgm}
	\end{subfigure}
	\caption{Experiments for the PnP-PGM algorithm \eqref{eq:PnP_it} on a synthetic image that is an approximate fixed point of a DRUNet denoiser. For each linear measurement operator, the theoretical linear convergence rate matches well the Log-L2 recovery curve.}
	\label{img:cartoon_pnp_pgm}
		\vspace{-5mm}
\end{figure}

\subsubsection{Natural Images}

In these experiments, we use a DRUNet denoiser trained on natural images (we used the weights provided by the DeepInverse library, see Acknowledgements at the end of the paper) and with the noise level $\eta$ as input (non-blind denoising). To obtain a fixed point of the denoiser, we apply  the denoiser on an original natural image multiple times with the entry noise level set to $\eta{=}0.18$ (This parameter is manually set to give good performance). For both images tested (see Figures~\ref{fig:starfish_pnp_pgm} and~\ref{fig:butterfly_pnp_pgm}), we observe that the theoretical linear convergence rate curve (calculated in the same way as the previous experiment) indeed matches well with the convergence curve $\|x_n - \hat{x}\|_2^2$. As expected, the higher blur (i.e. a degraded Restricted Isometry Constant) leads to a decrease in the performance of recovery. Also note that when the linear rate of convergence decreases as in the high blur experiment of Figure~\ref{fig:butterfly_pnp_pgm}, it becomes harder to distinguish from a fast sub-linear rate ($\frac{1}{n^2}$).

\begin{figure}[]
	\centering
	\begin{subfigure}[b]{0.3\linewidth}
		\centering
		\includegraphics[width=0.99\linewidth]{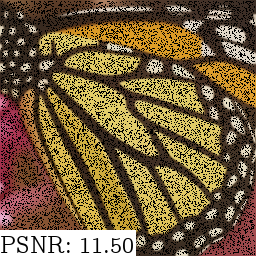}
		\caption{$y$ inpainting}
		\label{fig:butterfly_pnp_pgm_y_mask}
	\end{subfigure}
	\begin{subfigure}[b]{0.3\linewidth}
		\centering
		\includegraphics[width=0.99\linewidth]{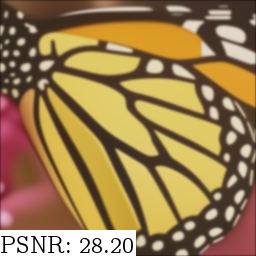}
		\caption{$y$ low blur $\sigma{=}1.0$}
		\label{fig:butterfly_pnp_pgm_y_blur_low}
	\end{subfigure}
	\begin{subfigure}[b]{0.3\linewidth}
		\centering
		\includegraphics[width=0.99\linewidth]{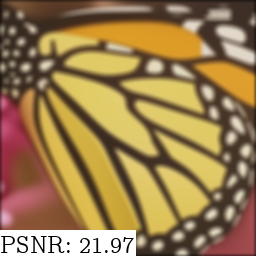}
		\caption{$y$ high blur $\sigma{=}3.0$}
		\label{fig:butterfly_pnp_pgm_y_blur_high}
	\end{subfigure}
	\\
	\begin{subfigure}[b]{0.3\linewidth}
		\centering
		\includegraphics[width=0.99\linewidth]{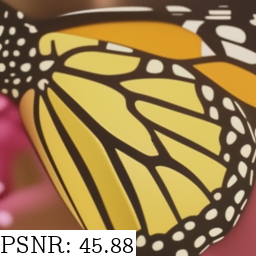}
		\caption{$x^*$ (inpainting)}
		\label{fig:butterfly_pnp_pgm_output_mask}
	\end{subfigure}
	\begin{subfigure}[b]{0.3\linewidth}
		\centering
		\includegraphics[width=0.99\linewidth]{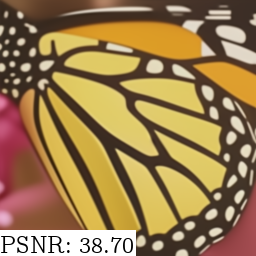}
		\caption{$x^*$ (low blur $\sigma=1.0$)}
		\label{fig:butterfly_pnp_pgm_output_blur_low}
	\end{subfigure}
	\begin{subfigure}[b]{0.3\linewidth}
		\centering
		\includegraphics[width=0.99\linewidth]{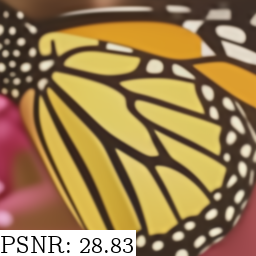}
		\caption{$x^*$ (high blur $\sigma=3.0$)}
		\label{fig:butterfly_pnp_pgm_output_blur_high}
	\end{subfigure}
	\\
	\begin{subfigure}[b]{0.325\linewidth}
		\centering
		\includegraphics[width=0.99\linewidth]{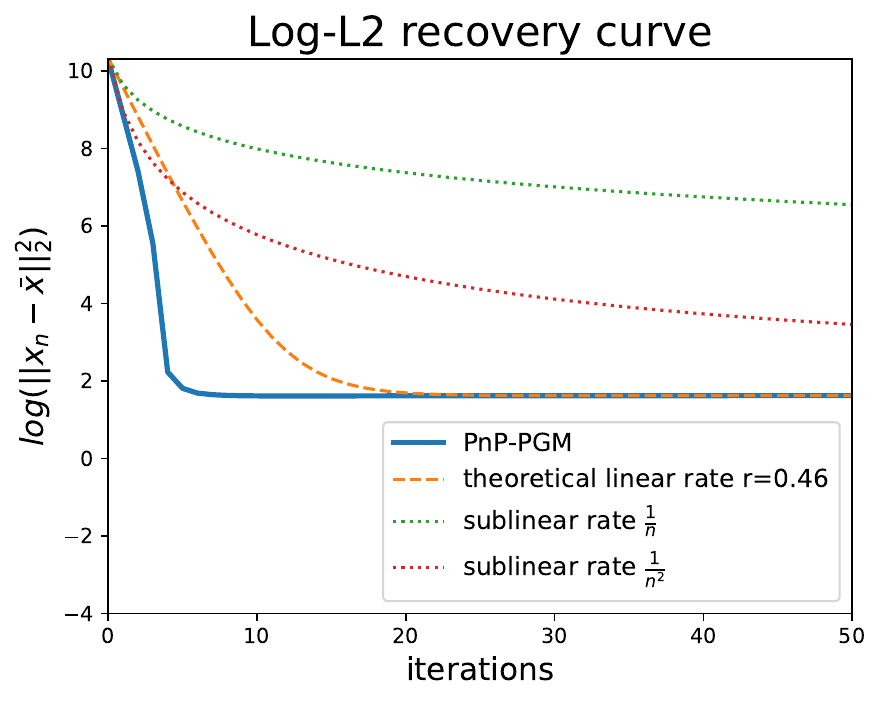}
		\caption{Log-L2 recovery curve (inpainting)}
		\label{plot:butterfly_pnp_pgm_mask}
	\end{subfigure}
	\begin{subfigure}[b]{0.325\linewidth}
		\centering
		\includegraphics[width=.99\linewidth]{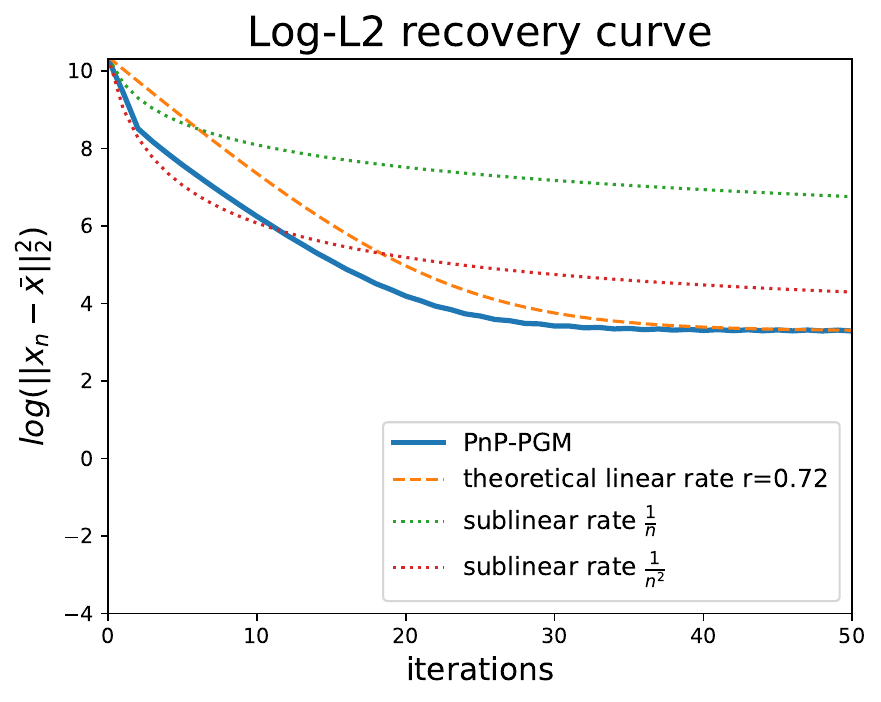}
		\caption{Log-L2 recovery curve (low blur $\sigma=1.0$)}
		\label{plot:butterfly_pnp_pgm_blur_low}
	\end{subfigure}
	\begin{subfigure}[b]{0.325\linewidth}
		\centering
		\includegraphics[width=0.99\linewidth]{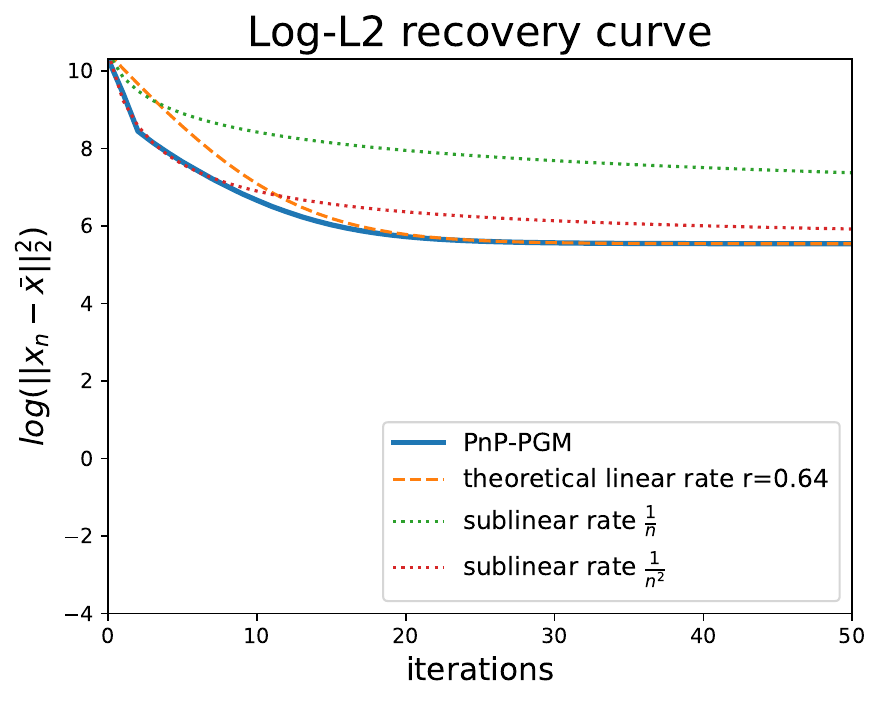}
		\caption{Log-L2 recovery curve (high blur $\sigma=3.0$)}
		\label{plot:butterfly_pnp_pgm_blur_high}
	\end{subfigure}
	\caption{Experiment of the PnP-PGM algorithm \eqref{eq:PnP_it} on an image that is an approximate fixed point of the blind denoiser. For each linear operator, the theoretical linear convergence rate matches well the Log-L2 recovery curve.}
	\label{fig:butterfly_pnp_pgm}
	\vspace*{-5mm}
\end{figure}

We verify in Figure~\ref{fig:verification_RIP_beta} that the quantities used to define the RIC ($\delta_n = \|(I-\mu A^TA)(D(x_n)-\hat{x})\|_2/\|D(x_n)-\hat{x}\|_2$) and the restricted Lipschitz constant ($\beta_n = \|D(x_n)-\hat{x}\|_2/\|x_n-\hat{x}\|$) are bounded over the iterates of PNP-PGM. In particular the condition on the rate criterion $\delta_n \beta_n <1$ is verified. The convergence to $1$ of $\delta_n\beta_n$ is to be expected for the same reason as in~\eqref{eq:numerical_linear_bound}. We  verify that the chosen denoiser maps to a set of approximate fixed points with a deviation of $2\%$ of the quantity $\|D^2(x_n)-D(x_n)\|_2$ with respect to $\|D(x_n)\|_2$. Hence the deviation from our theoretical set-up is low.

\begin{figure}
	\centering
	\begin{subfigure}[b]{0.45\linewidth}
		\centering
		\includegraphics[width=0.99\linewidth]{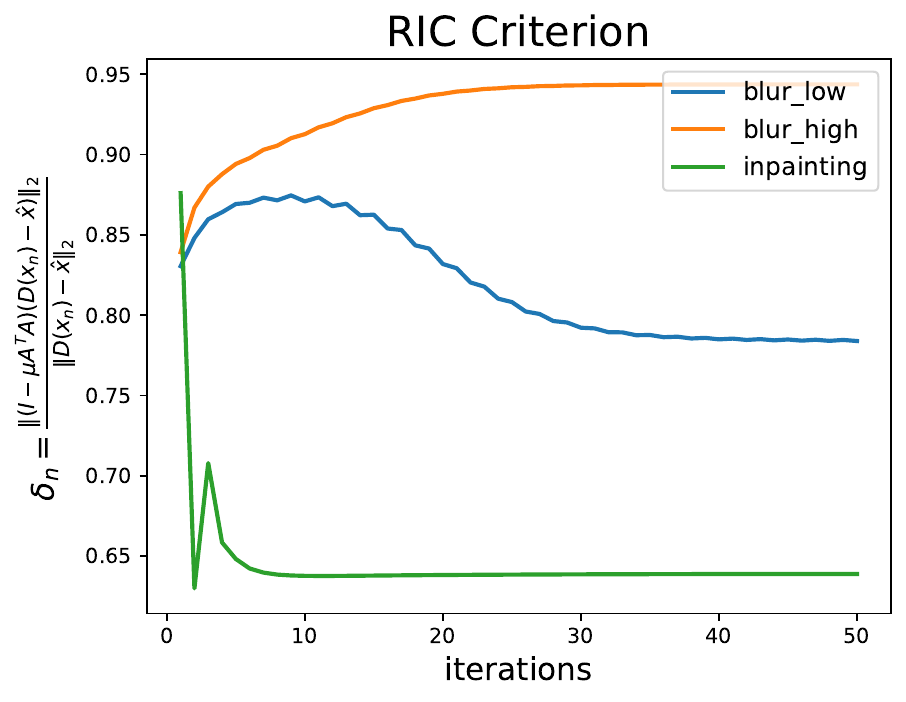}
		\caption{Evolution of the RIC criterion \\ 
		\hfill
		}
		\label{fig:butterfly_pnp_pgm_rip}
	\end{subfigure}
	\begin{subfigure}[b]{0.45\linewidth}
		\centering
		\includegraphics[width=0.99\linewidth]{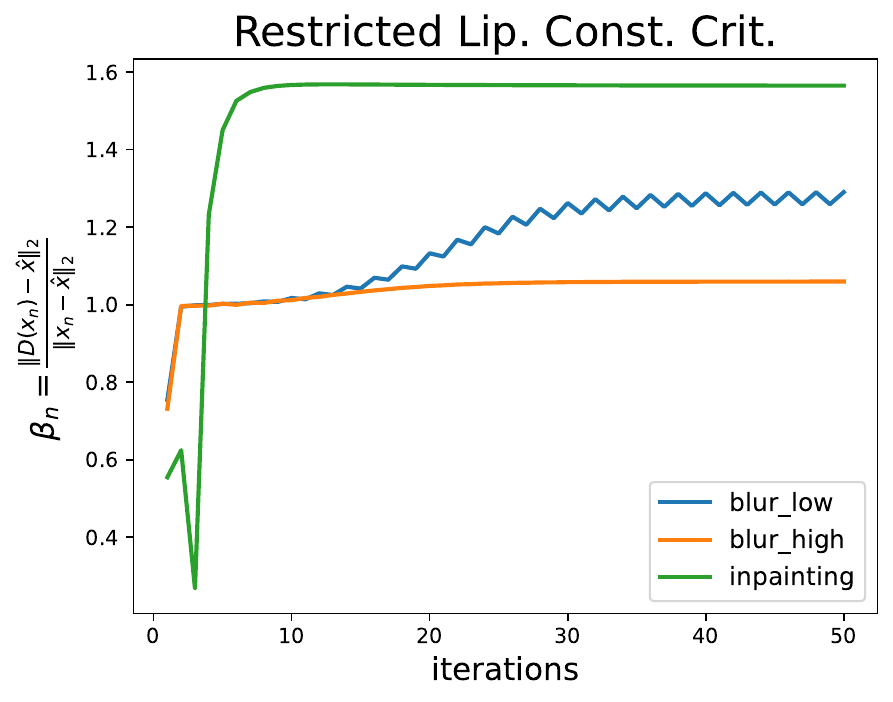}
		\caption{Evolution of the restricted Lipschitz constant criterion}
		\label{fig:butterfly_pnp_pgm_beta}
	\end{subfigure}\\
	\begin{subfigure}[b]{0.45\linewidth}
		\centering
		\includegraphics[width=0.99\linewidth]{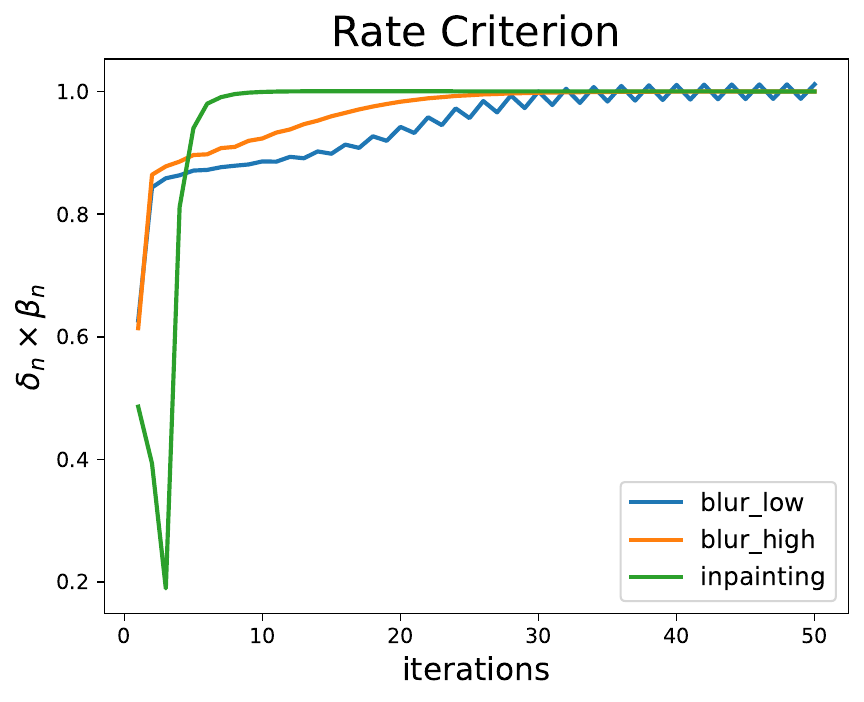}
		\caption{Evolution of the rate criterion\\
			\hfill
		}
		\label{fig:butterfly_pnp_pgm_linear_rate}
	\end{subfigure}
	\begin{subfigure}[b]{0.45\linewidth}
		\centering
		\includegraphics[width=0.99\linewidth]{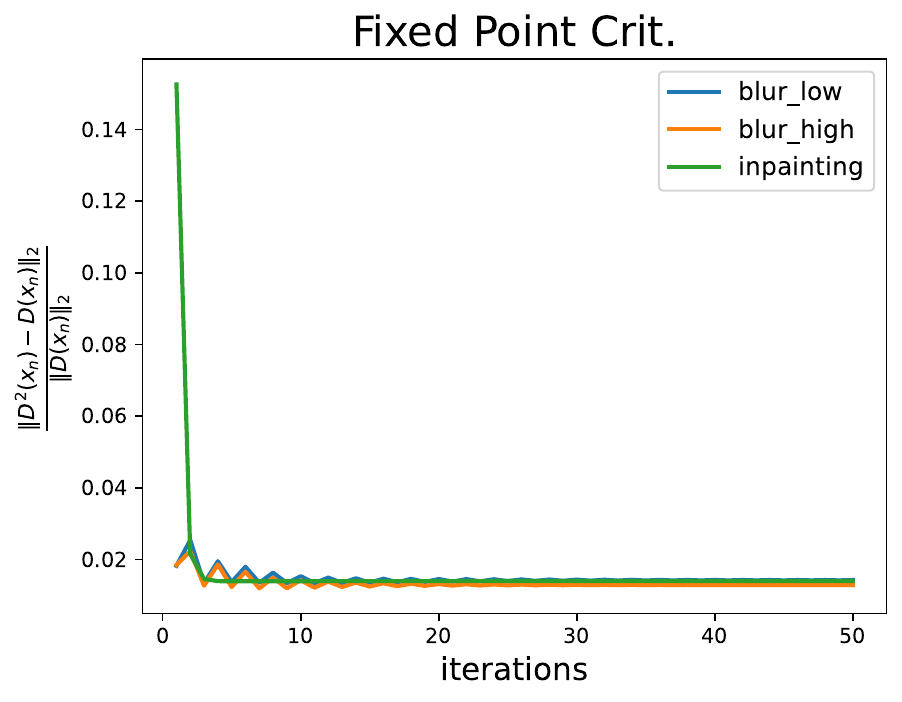}
		\caption{Fixed point evaluation of the denoiser $D_{\Sigma}$ during the iterations}
		\label{fig:butterfly_pnp_pgm_idem}
	\end{subfigure}
	\caption{ Evaluation of the hypotheses of Theorem \ref{th:gen_convevergence} within the PnP-PGM iterations for the Butterfly image (Figure \ref{fig:butterfly_pnp_pgm}). We assume that $D(x_n), \hat{x} \in \Sigma$. (\ref{fig:butterfly_pnp_pgm_rip}, \ref{fig:butterfly_pnp_pgm_beta}): The RIP property is verified with $\delta_n< 1$ and $D$ has a bounded restricted Lipschitz constant $\beta_n$. The rate criterion $\delta_n \times \beta_n$ verifies  $\delta_n \times \beta_n < 1$ (\ref{fig:butterfly_pnp_pgm_linear_rate}).(\ref{fig:butterfly_pnp_pgm_idem}): The denoiser $D$ is approximately idempotent, with less than $0.02$ relative error over the iterations. }
	\label{fig:verification_RIP_beta}
\end{figure}

\begin{figure}[]
	\centering
	\begin{subfigure}[b]{0.3\linewidth}
		\centering
		\includegraphics[width=0.99\linewidth]{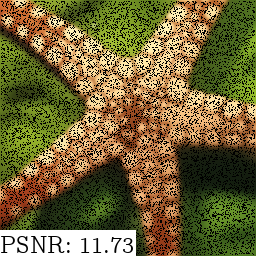}
		\caption{$y$ inpainting}
		\label{fig:starfish_pnp_pgm_y_mask}
	\end{subfigure}
	\begin{subfigure}[b]{0.3\linewidth}
		\centering
		\includegraphics[width=0.99\linewidth]{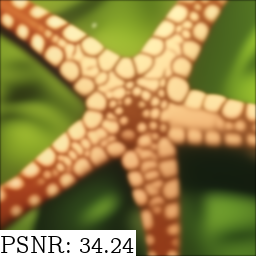}
		\caption{$y$ low blur $\sigma = 1.0$}
		\label{fig:starfish_pnp_pgm_y_blur_low}
	\end{subfigure}
	\begin{subfigure}[b]{0.3\linewidth}
		\centering
		\includegraphics[width=0.99\linewidth]{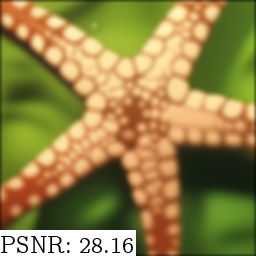}
		\caption{$y$ high blur $\sigma = 1.0$}
		\label{fig:starfish_pnp_pgm_y_blur_high}
	\end{subfigure}\\
	\begin{subfigure}[b]{0.3\linewidth}
		\centering
		\includegraphics[width=0.99\linewidth]{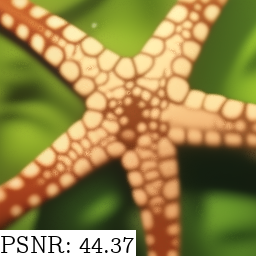}
		\caption{$x^*$ (inpainting)}
		\label{fig:starfish_pnp_pgm_output_mask}
	\end{subfigure}
	\begin{subfigure}[b]{0.3\linewidth}
		\centering
		\includegraphics[width=0.99\linewidth]{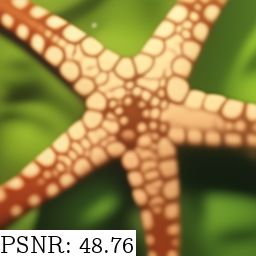}
		\caption{$x^*$ (low blur $\sigma=1.0$)}
		\label{fig:starfish_pnp_pgm_output_blur_low}
	\end{subfigure}
	\begin{subfigure}[b]{0.3\linewidth}
		\centering
		\includegraphics[width=0.99\linewidth]{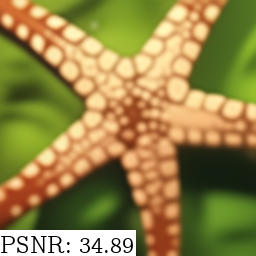}
		\caption{$x^*$ (high blur $\sigma=3.0$)}
		\label{fig:starfish_pnp_pgm_output_blur_high}
	\end{subfigure}\\
	\begin{subfigure}[b]{0.325\linewidth}
		\centering
		\includegraphics[width=0.99\linewidth]{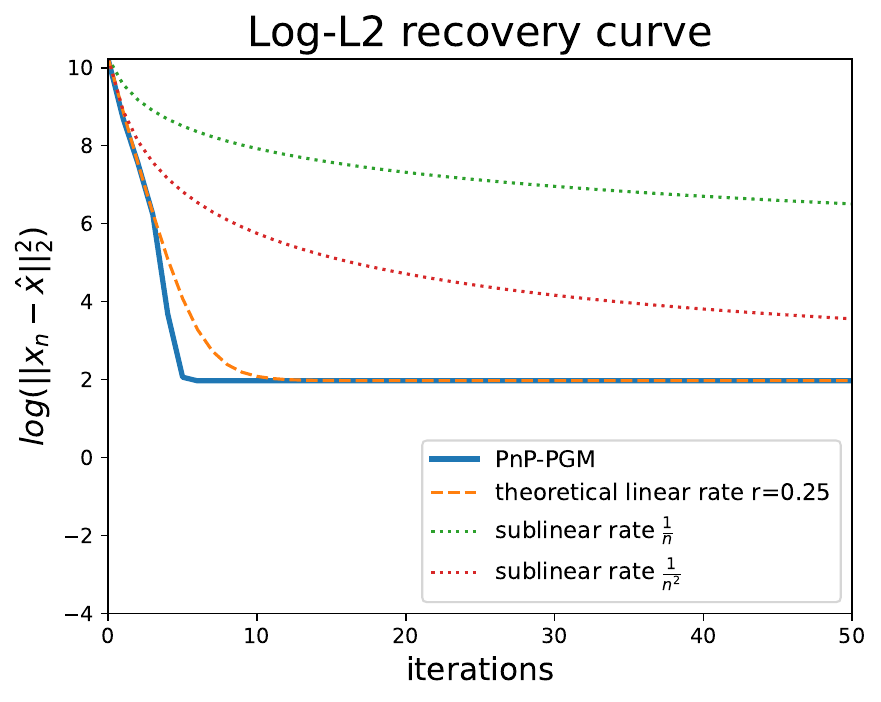}
		\caption{Log-L2 recovery curve (inpainting)}
		\label{plot:starfish_pnp_pgm_mask}
	\end{subfigure}
	\begin{subfigure}[b]{0.325\linewidth}
		\centering
		\includegraphics[width=.99\linewidth]{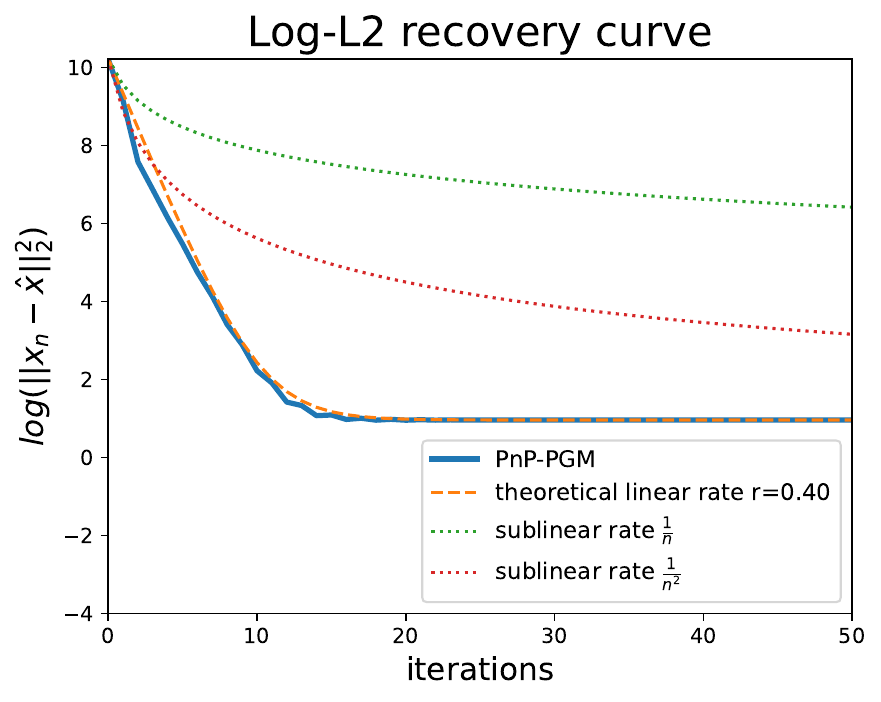}
		\caption{Log-L2 recovery curve (low blur $\sigma=1.0$)}
		\label{plot:starfish_pnp_pgm_blur_low}
	\end{subfigure}
	\begin{subfigure}[b]{0.325\linewidth}
		\centering
		\includegraphics[width=0.99\linewidth]{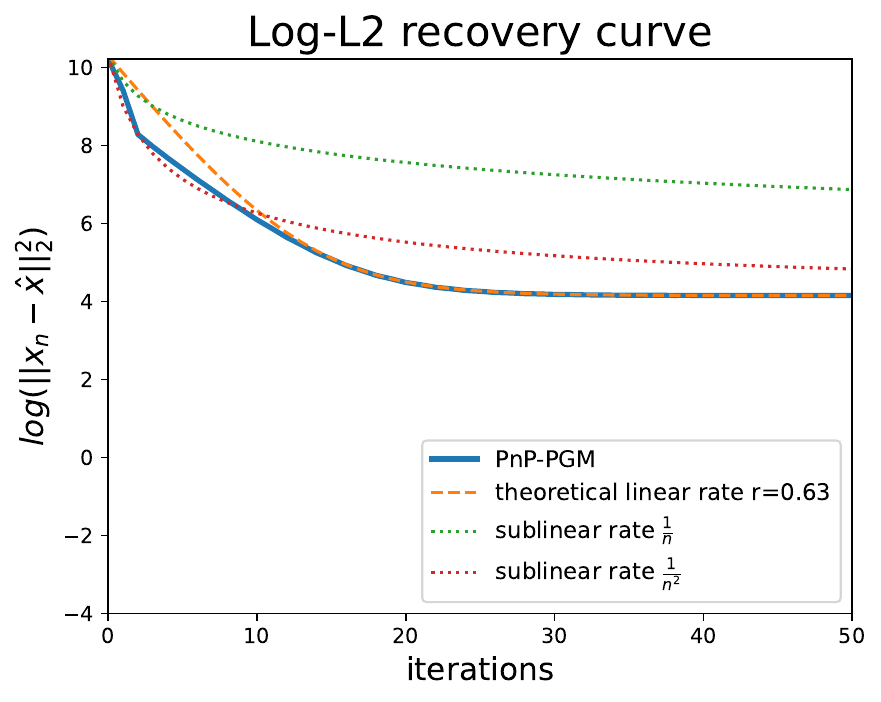}
		\caption{Log-L2 recovery curve (high blur $\sigma=3.0$)}
		\label{plot:starfish_pnp_pgm_blur_high}
	\end{subfigure}
	\caption{Experiment of the PnP-PGM algorithm \eqref{eq:PnP_it} on an image that is an approximate fixed point of the blind denoiser. For each linear operator,  the theoretical linear convergence rate matches well the Log-L2 recovery curve.}
	\label{fig:starfish_pnp_pgm}
		\vspace{-5mm}
\end{figure}

\subsubsection{Comparison between PNP-PGM and GM-RED}\label{sec:comp_PGM_RED}
We perform the same previous experiments with the GM-RED algorithm, which falls outside the theoretical results provided in this paper. Once again, we observe that the theoretical linear convergence rate matched very well the convergence curve $\|x_n-\hat{x}\|_2^2$ (see Fig. \ref{img:images_synthetic_gm_red}, \ref{img:images_butterfly_gm_red}, \ref{img:images_starfish_gm_red}). Furthermore, the measured convergence rate of GM-RED is slower than PnP-PGM, reinforcing our initial conjecture. Interestingly, we found that GM-RED was able to obtain a better recovery of our estimated fixed point than PnP-PGM in some experiments (particularly when the measurement operator is better conditioned (low blur)). This could be due to the additional parameter in GM-RED which allows a finer direction towards the low-dimensional model $\Sigma$ induced by the denoiser or simply that our choice of approximate fixed point just matches better the GM-RED algorithm. This opens a question to include approximation error in the projection  in our analysis (as was for example done in~\cite{golbabaee2018inexact}).

While these experiments do not prove the optimality (with respect to the rate of convergence) of the  projected gradient method for low-dimensional recovery, it suggests that our analysis lays out good foundations for a global understanding of low-dimensional recovery with deep priors.


\begin{figure}[]
	\centering
	\begin{subfigure}[b]{0.325\linewidth}
		\centering
		\includegraphics[width=0.99\linewidth]{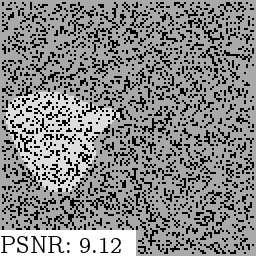}
		\caption{$y$ inpainting}
		\label{fig:synthetic_gm_red_y_mask}
	\end{subfigure}
	\begin{subfigure}[b]{0.325\linewidth}
		\centering
		\includegraphics[width=0.99\linewidth]{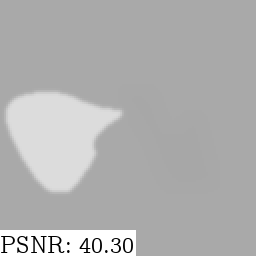}
		\caption{$y$ low blur $\sigma=1.0$}
		\label{fig:synthetic_gm_red_y_blur_low}
	\end{subfigure}
	\begin{subfigure}[b]{0.325\linewidth}
		\centering
		\includegraphics[width=0.99\linewidth]{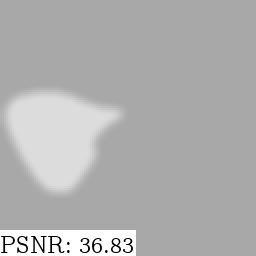}
		\caption{$y$ high blur $\sigma=3.0$}
		\label{fig:synthetic_gm_red_y_blur_high}
	\end{subfigure}\\
	\begin{subfigure}[b]{0.325\linewidth}
		\centering
		\includegraphics[width=0.99\linewidth]{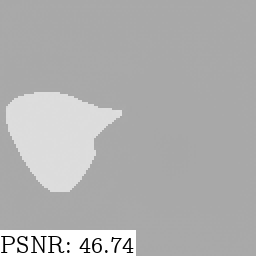}
		\caption{$x^\ast$ (inpainting)}
		\label{fig:synthetic_gm_red_output_mask}
	\end{subfigure}
	\begin{subfigure}[b]{0.325\linewidth}
		\centering
		\includegraphics[width=0.99\linewidth]{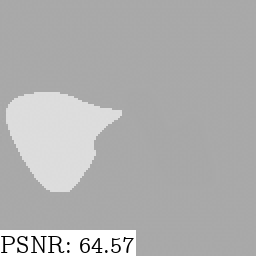}
		\caption{$x^\ast$ (low blur $\sigma{=}1.0$)}
		\label{fig:synthetic_gm_red_output_blur_low}
	\end{subfigure}
	\begin{subfigure}[b]{0.325\linewidth}
		\centering
		\includegraphics[width=0.99\linewidth]{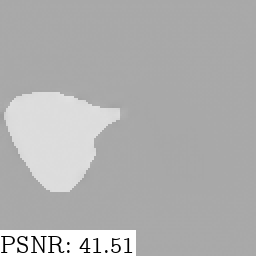}
		\caption{$x^\ast$ (high blur $\sigma=3.0$)}
		\label{fig:synthetic_gm_red_output_blur_high}
	\end{subfigure}\\
	\begin{subfigure}[b]{0.325\linewidth}
		\centering
		\includegraphics[width=0.99\linewidth]{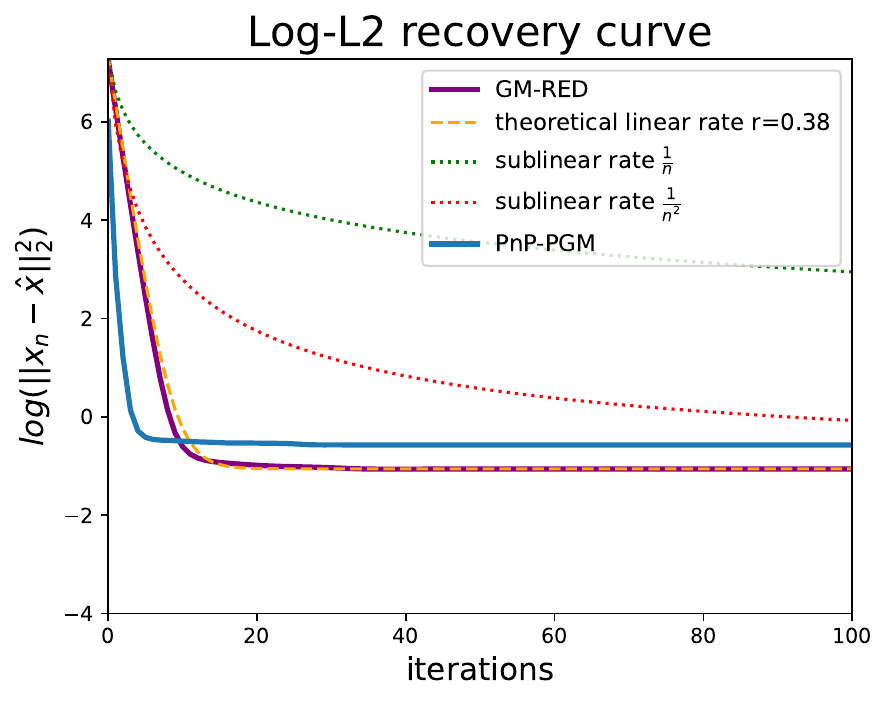}
		\caption{Log-L2 recovery curve (inpainting)}
		\label{plot:synthetic_gm_red_mask}
	\end{subfigure}
	\begin{subfigure}[b]{0.325\linewidth}
		\centering
		\includegraphics[width=.99\linewidth]{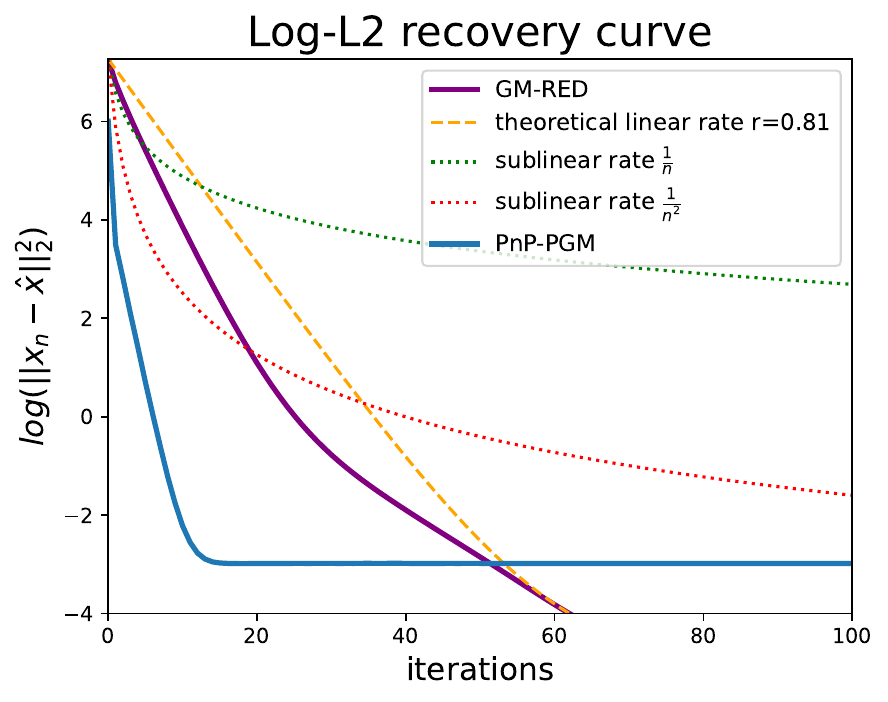}
		\caption{Log-L2 recovery curve (low blur $\sigma=1.0$)}
		\label{plot:synthetic_gm_red_blur_low}
	\end{subfigure}
	\begin{subfigure}[b]{0.325\linewidth}
		\centering
		\includegraphics[width=0.99\linewidth]{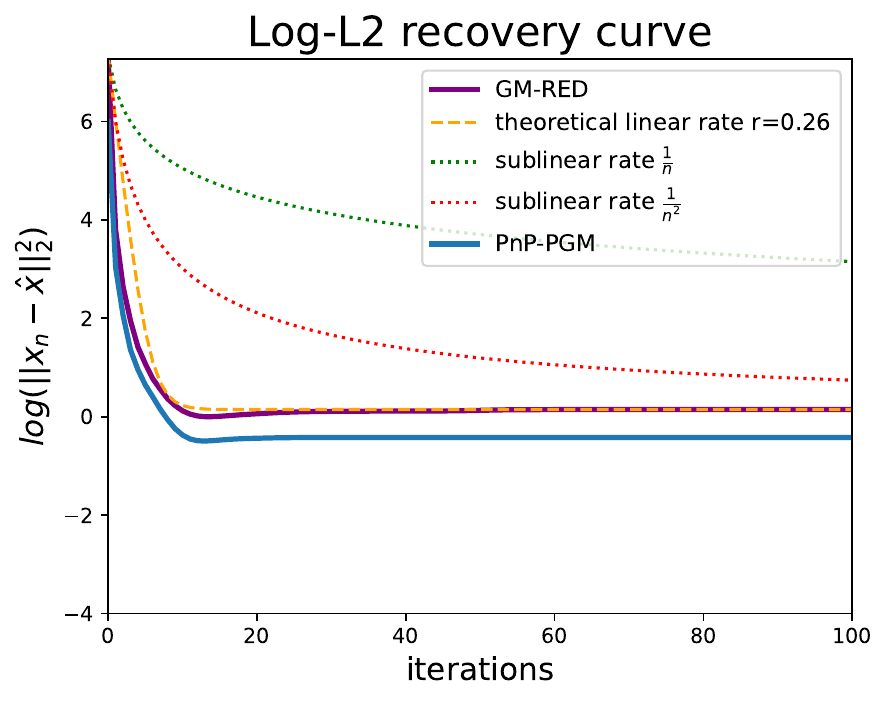}
		\caption{Log-L2 recovery curve (high blur $\sigma=3.0$)}
		\label{plot:synthetic_gm_red_blur_high}
	\end{subfigure}
	\caption{Experiments of the GM-RED algorithm \eqref{eq:PnP_it} on a synthetic image for different measurement operators. We observe that GM-RED presents a linear convergence rate. Moreover, the measured convergence rate is slower than PnP-PGM.}
	\label{img:images_synthetic_gm_red}
	\vspace{-3mm}
		\vspace{-5mm}
\end{figure}

\begin{figure}[]
	\centering
	\begin{subfigure}[b]{0.325\linewidth}
		\centering
		\includegraphics[width=0.99\linewidth]{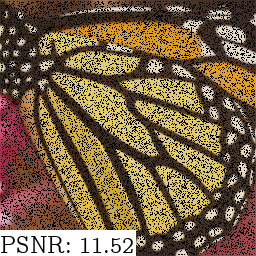}
		\caption{}
		\label{fig:butterfly_gm_red_y_mask}
	\end{subfigure}
	\begin{subfigure}[b]{0.325\linewidth}
		\centering
		\includegraphics[width=0.99\linewidth]{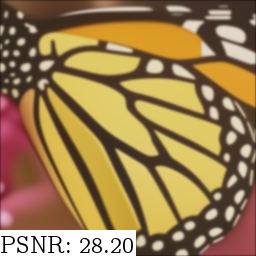}
		\caption{}
		\label{fig:butterfly_gm_red_y_blur_low}
	\end{subfigure}
	\begin{subfigure}[b]{0.325\linewidth}
		\centering
		\includegraphics[width=0.99\linewidth]{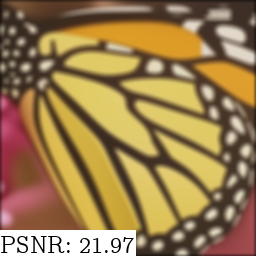}
		\caption{}
		\label{fig:butterfly_gm_red_y_blur_high}
	\end{subfigure}
	\\
	\begin{subfigure}[b]{0.325\linewidth}
		\centering
		\includegraphics[width=0.99\linewidth]{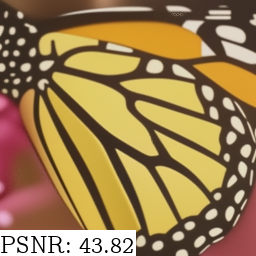}
		\caption{}
		\label{fig:butterfly_gm_red_output_mask}
	\end{subfigure}
	\begin{subfigure}[b]{0.325\linewidth}
		\centering
		\includegraphics[width=0.99\linewidth]{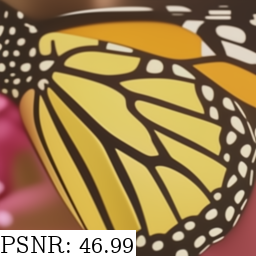}
		\caption{}
		\label{fig:butterfly_gm_red_output_blur_low}
	\end{subfigure}
	\begin{subfigure}[b]{0.325\linewidth}
		\centering
		\includegraphics[width=0.99\linewidth]{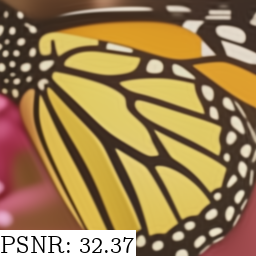}
		\caption{}
		\label{fig:butterfly_gm_red_output_blur_high}
	\end{subfigure}
	\\
	\begin{subfigure}[b]{0.325\linewidth}
		\centering
		\includegraphics[width=0.99\linewidth]{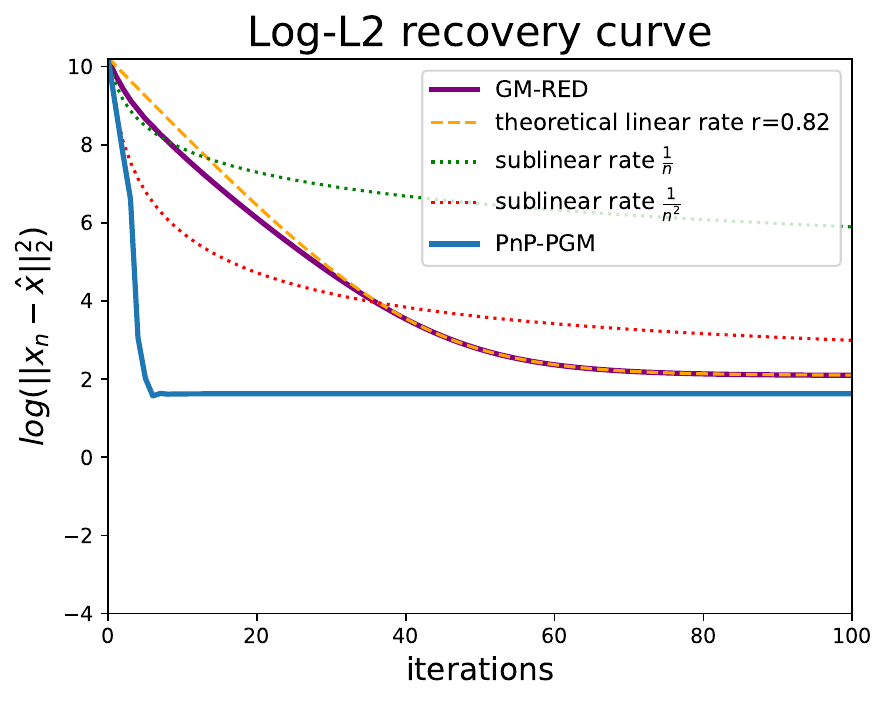}
		\caption{}
		\label{plot:butterfly_gm_red_mask}
	\end{subfigure}
	\begin{subfigure}[b]{0.325\linewidth}
		\centering
		\includegraphics[width=.99\linewidth]{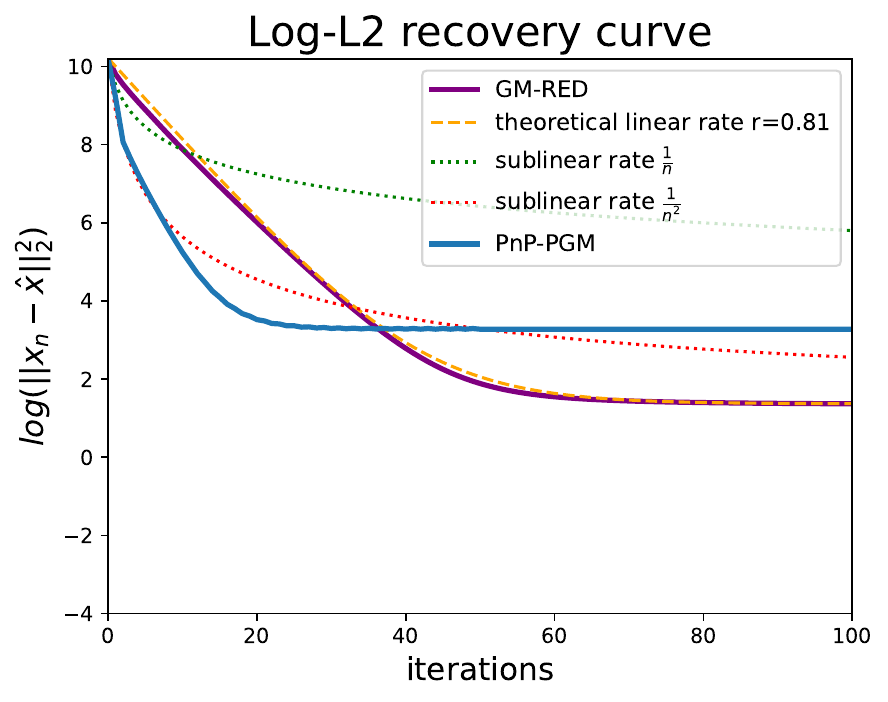}
		\caption{}
		\label{plot:butterfly_gm_red_blur_low}
	\end{subfigure}
	\begin{subfigure}[b]{0.325\linewidth}
		\centering
		\includegraphics[width=0.99\linewidth]{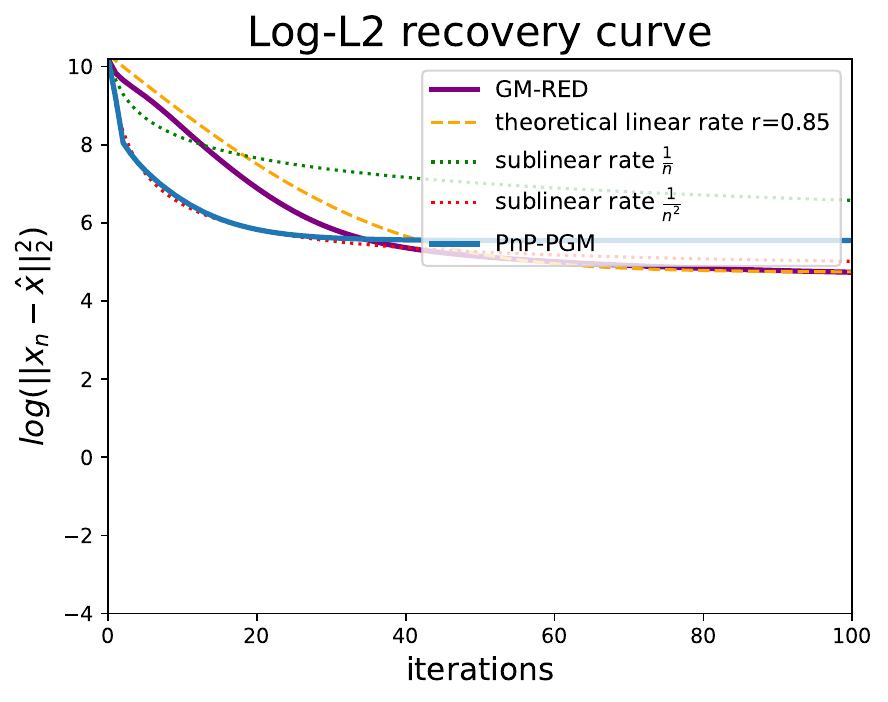}
		\caption{}
		\label{plot:butterfly_gm_red_blur_high}
	\end{subfigure}
	\caption{Experiment of the GM-RED algorithm \eqref{eq:PnP_it} on the butterfly image for different linear operations. We observe that GM-RED presents a linear convergence rate. Moreover, the measured convergence rate is slower than PnP-PGM.}
	\label{img:images_butterfly_gm_red}
		\vspace*{-4mm}
\end{figure}

\begin{figure}[ht!]
	\centering
	\begin{subfigure}[b]{0.325\linewidth}
		\centering
		\includegraphics[width=0.99\linewidth]{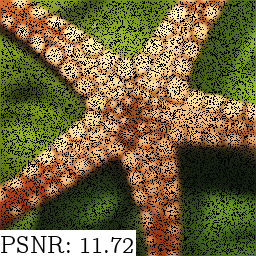}
		\caption{$y$ (inpainting)}
		\label{fig:starfish_gm_red_y_mask}
	\end{subfigure}
	\begin{subfigure}[b]{0.325\linewidth}
		\centering
		\includegraphics[width=0.99\linewidth]{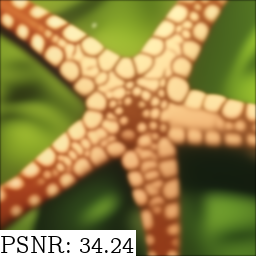}
		\caption{$y$ (low blur $\sigma=1.0$)}
		\label{fig:starfish_gm_red_y_blur_low}
	\end{subfigure}
	\begin{subfigure}[b]{0.325\linewidth}
		\centering
		\includegraphics[width=0.99\linewidth]{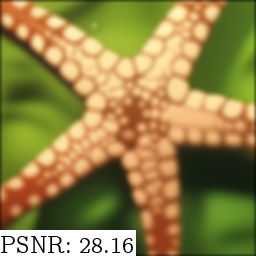}
		\caption{$y$ (high blur $\sigma=3.0$)}
		\label{fig:starfish_gm_red_y_blur_high}
	\end{subfigure}
	\\
	\begin{subfigure}[b]{0.325\linewidth}
		\centering
		\includegraphics[width=0.99\linewidth]{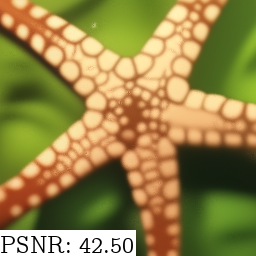}
		\caption{$x^\ast$ (inpainting)}
		\label{fig:starfish_gm_red_output_mask}
	\end{subfigure}
	\begin{subfigure}[b]{0.325\linewidth}
		\centering
		\includegraphics[width=0.99\linewidth]{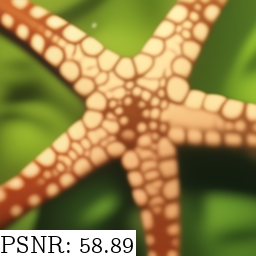}
		\caption{$x^\ast$ (low blur $\sigma=1.0$)}
		\label{fig:starfish_gm_red_output_blur_low}
	\end{subfigure}
	\begin{subfigure}[b]{0.325\linewidth}
		\centering
		\includegraphics[width=0.99\linewidth]{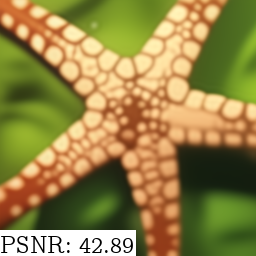}
		\caption{$x^\ast$ (high blur $\sigma=3.0$)}
		\label{fig:starfish_gm_red_output_blur_high}
	\end{subfigure}
	\\
	\begin{subfigure}[b]{0.325\linewidth}
		\centering
		\includegraphics[width=0.99\linewidth]{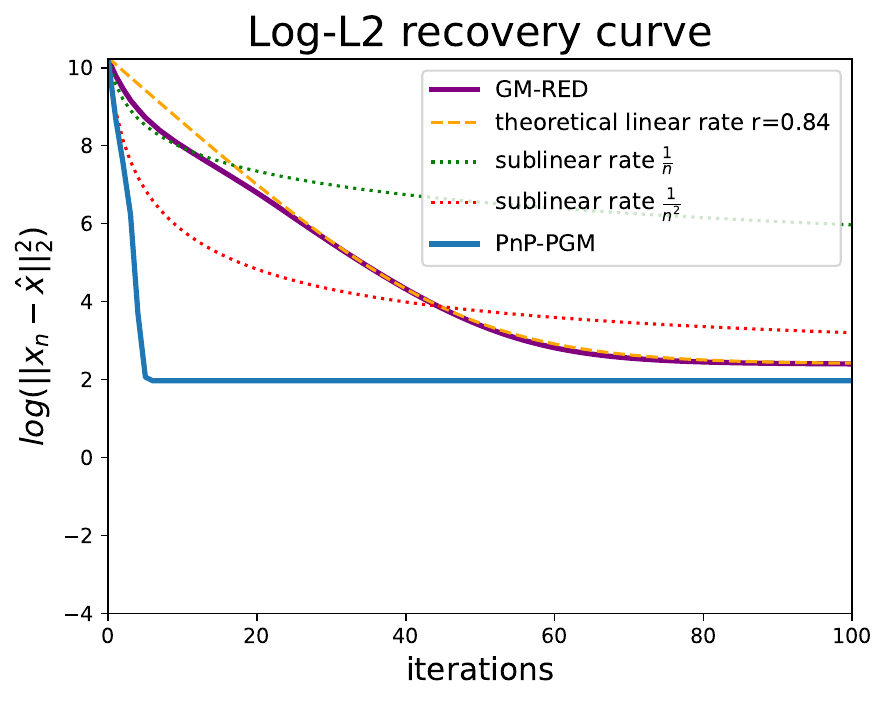}
		\caption{}
		\label{plot:starfish_gm_red_mask}
	\end{subfigure}
	\begin{subfigure}[b]{0.325\linewidth}
		\centering
		\includegraphics[width=.99\linewidth]{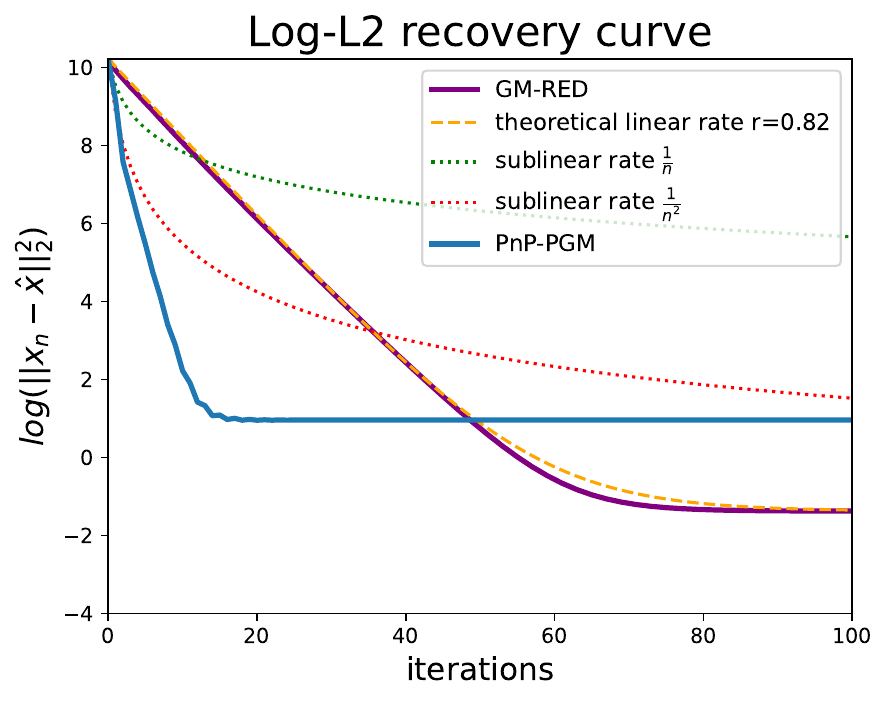}
		\caption{}
		\label{plot:starfish_gm_red_blur_low}
	\end{subfigure}
	\begin{subfigure}[b]{0.325\linewidth}
		\centering
		\includegraphics[width=0.99\linewidth]{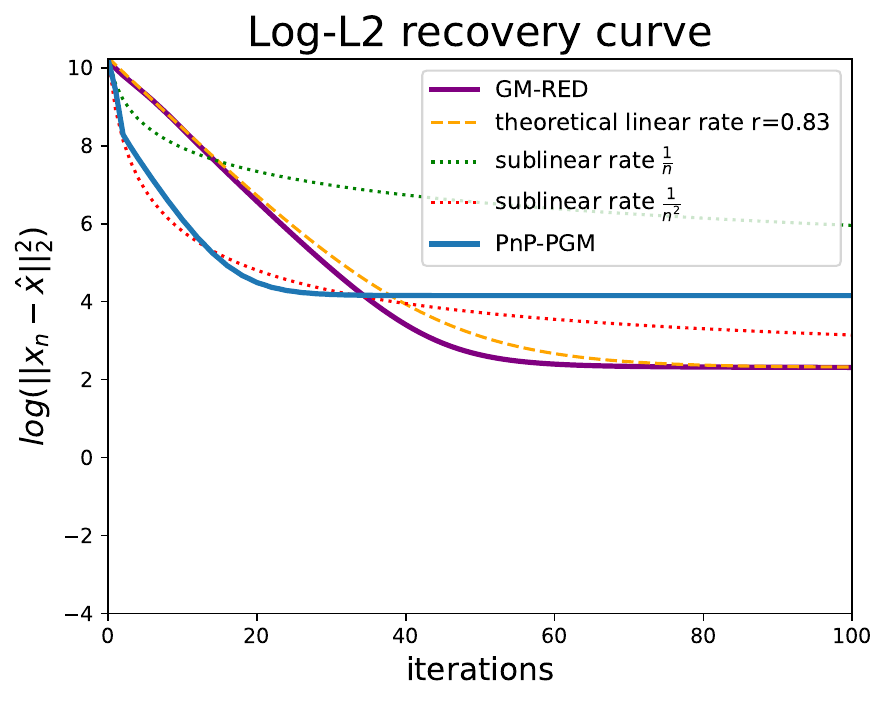}
		\caption{}
		\label{plot:starfish_gm_red_blur_high}
	\end{subfigure}
	\caption{Experiment of the GM-RED algorithm \eqref{eq:PnP_it} on the starfish image for different linear operations. We observe that GM-RED presents a linear convergence rate. Moreover, the measured convergence rate is slower than PnP-PGM.}
	\label{img:images_starfish_gm_red}
		\vspace{-5mm}
\end{figure}

\section{Conclusion} \label{sec:conclusion}
We have given a convergence analysis of a class of projected descent algorithms for the recovery of low-dimensional models. Our result explicitly quantifies the convergence rate with the restricted isometry constants of the measurement operator and a newly introduced restricted Lipschitz condition on the operator projecting onto the model set. This decouples the role of the geometry of the model and the quality of the measurement operator in the rate of convergence.

More particularly the orthogonal projection yields very general guarantees for  general sets. These guarantees can be improved in the case of sparse recovery and iterative hard thresholding, showing that hard thresholding is indeed optimal for the convergence rate (via the restricted Lipschitz constant) when considering the whole class of sparse models (for any sparsity).

Our work lays out the foundation of a theoretical framework for optimal algorithms for the recovery of low-dimensional beyond the variational approach. Many ideas can be explored to generalize this work. Extending to more general classes of algorithms, exploring the tightness of our different results, or studying the impact of the noise in the search for optimality are possible interesting leads. Another possibility would be to add more flexibility to the uniform rate condition and to consider a finite time ''burning'' period, i.e. linear convergence guaranteed after a fixed number of iterations.

Our  results guarantee linear convergence for solving inverse problems with deep priors. They also raise the question of learning a projection (a denoiser in the plug-and-play framework) with a good \emph{restricted} Lipschitz constant, thus relaxing the global Lipschitz condition.

\section{Acknowledgements}\label{sec:ack}
Experiments presented in this paper were carried out using the PlaFRIM experimental testbed, supported by Inria, CNRS (LABRI and IMB), Université de Bordeaux, Bordeaux INP and Conseil Régional d’Aquitaine (see \url{https://www.plafrim.fr}).
Furthermore, we are grateful to the DeepInverse python library (\url{https://deepinv.github.io/deepinv/index.html}) from which the code and weights of the denoiser for natural images was taken from. This work was supported by the French National
Research Agency (ANR) under reference ANR-20-CE40-0001 (EFFIREG project), and by PEPR PDE\_AI. 
We thank the anonymous reviewers whose comments helped improve this article.

\bibliographystyle{elsarticle-num}
\bibliography{towards_optimal_algorithms.bib}
\end{document}